%% file: main.tex
\def\eprintversion{1}
\ifnum\eprintversion=0
	\documentclass{llncs}
\else
	\documentclass{llncs}[11pt]
	\usepackage{fullpage}
\fi

\newif\ifsubmission
\submissionfalse

\newif\iffull
\fulltrue

\iffull
	\makeatletter
	\def\tableofcontents{\section*{\contentsname\@mkboth{{\contentsname}}%
	{{\contentsname}}}
	 \def\authcount##1{\setcounter{auco}{##1}\setcounter{@auth}{1}}
	 \def\lastand{\ifnum\value{auco}=2\relax
	                 \unskip{} \andname\
	              \else
	                 \unskip \lastandname\
	              \fi}%
	 \def\and{\stepcounter{@auth}\relax
	          \ifnum\value{@auth}=\value{auco}%
	             \lastand
	          \else
	             \unskip,
	          \fi}%
	 \@starttoc{toc}\if@restonecol\twocolumn\fi}
	\makeatother
\if

\usepackage[dvips]{graphicx} 
\usepackage{enumerate}
\usepackage{amsmath}
\usepackage{amsfonts}
\usepackage{amssymb}
\usepackage{color} 
\usepackage{multicol}
\usepackage{mathabx}
\usepackage{calc}
\ifsubmission
	\usepackage[
	pdftitle={Nonadaptive One-Way to Hiding Implies Adaptive Quantum Reprogramming},
    pdfauthor={},colorlinks=true,linkcolor=blue,citecolor=blue]{hyperref}
\else
	\usepackage[
	pdftitle={Nonadaptive One-Way to Hiding Implies Adaptive Quantum Reprogramming},
    pdfauthor={Joseph Jaeger},
    colorlinks=true,
    linkcolor=blue,
    citecolor=blue,
	urlcolor=blue]{hyperref}
\fi
\usepackage{colonequals}  
\usepackage{cleveref}
\usepackage{xcolor}
\usepackage{wrapfig}
\input{macros.tex}

\usepackage{thm-restate}
\usepackage{csquotes}

\usepackage{physics}
\newif\ifnotes
\notesfalse
\ifnotes
\newcommand{\js}[1]{$\ll$\textsf{\color{blue} JS: { #1}}$\gg$}
\else
\newcommand{\js}[1]{}
\fi

\begin{document}

\title{Nonadaptive One-Way to Hiding Implies\texorpdfstring{\\}{ }Adaptive Quantum Reprogramming}
\titlerunning{Nonadaptive One-Way to Hiding Implies Adaptive Quantum Reprogramming}
%

	\renewcommand{\orcidID}[1]{\kern .08em\href{https://orcid.org/#1}{\protect\includegraphics[keepaspectratio,height=1.0em]{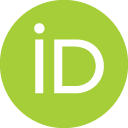}}}
	\author{Joseph Jaeger\orcidID{0000-0002-4934-3405}}
	\authorrunning{J. Jaeger}
	\institute{
			School of Cybersecurity and Privacy\\
			Georgia Institute of Technology\\
			Atlanta, Georgia, USA\\
			\email{josephjaeger@gatech.edu}\\
			\href{https://cc.gatech.edu/~josephjaeger/}{cc.gatech.edu/\textasciitilde josephjaeger/}
	}


\begin{center}
	A preliminary version of this work appears in the \textsc{Asiacrypt 2025} proceedings (\copyright IACR 2025). \newline
	This is the full version.
\end{center}
\vspace{1em}
\iffull
	\renewcommand*\contentsname{Contents} 
	{\def\addcontentsline#1#2#3{}{\let\newpage\relax\maketitle}} 
\fi

\noindent
\makebox[\linewidth]{\today}

\input{abstract}

{\hypersetup{linkcolor=black}\tableofcontents}
\newpage

\newcommand{\changelocaltocdepth}[1]{%
  \addtocontents{toc}{\protect\setcounter{tocdepth}{#1}}%
  \setcounter{tocdepth}{#1}%
}

\input{intro}

\input{prelims}

\input{panzeng}
\input{ghhm}

\input{unruh}

\input{other-frameworks}

\input{arxiv.bbl}
\appendix
\addcontentsline{toc}{section}{Appendices}
\section*{Change Log}
	\begin{itemize}
		\item May 2024: Fixed recurring typo pointed out by Hans Heum
		\item November 2025: Asiacrypt full version with small changes throughout to add clarity
	\end{itemize}
\input{ghhm-details}

\end{document}

%% file: macros.tex
\definecolor{rvone}{RGB}{139,34,82}
\definecolor{rvtwo}{rgb}{0.00,0.62,0.22}

\renewcommand{\epsilon}{\varepsilon}

\renewcommand{\Pr}{\mathsf{Pr}}

\newcommand{\Adv}{\mathsf{Adv}}

\newcommand{\abs}[1]{\left| #1 \right|}

\newcommand{\sG}{{\normalfont \textsf{G}}}
\newcommand{\sH}{{\normalfont \textsf{H}}}

\newcommand{\bD}{\mathbf{D}}

\newcommand{\advA}{\mathcal{A}}
\newcommand{\advB}{\mathcal{B}}
\newcommand{\advD}{\mathcal{D}}
\newcommand{\env}{\mathcal{E}}
\newcommand{\distD}{\bD}

\makeatletter
\newcommand\parag{%
   \@startsection{paragraph}{4}{\z@}%
       {-3\p@ \@plus -1\p@ \@minus -1\p@}%
       {-0.2em \@plus -0.12em \@minus -0.05em}%
       {\normalfont\normalsize\scshape}}
\newcommand{\heading}[1]{\parag{#1}}
\makeatother 
\renewcommand{\heading}[1]{\subsubsection{#1.}}

\newcommand{\ind}{\hspace*{10pt}}
\newcommand{\comment}[1]{\hspace{5pt}{\small\color{blue} /$\!\!$/\ #1}}

\newcommand{\true}{\texttt{true}}

\newcommand{\sett}[1]{\{#1\}}
\newcommand{\getsr}{{\:{\leftarrow{\hspace*{-3pt}\raisebox{.75pt}{$\scriptscriptstyle\$$}}}\:}}
\newcommand{\bits}{\{0,1\}}
\newcommand{\xor}{\oplus}

\definecolor{highlight-gray}{gray}{0.90}
\definecolor{highlight-yellow}{cmyk}{0,0,0.90,0}
\definecolor{highlight-cyan}{cmyk}{0.40,0,0,0}
\definecolor{highlight-magenta}{cmyk}{0,0.90,0,0}

\newcommand{\gamechange}[2][highlight-gray]{{\setlength{\fboxsep}{0pt}\colorbox{#1}{\ifmmode$\displaystyle#2$\else#2\fi}}}
\newcommand{\gamechangey}[2][highlight-yellow]{{\setlength{\fboxsep}{0pt}\colorbox{#1}{\ifmmode$\displaystyle#2$\else#2\fi}}}
\newcommand{\gamechangec}[2][highlight-cyan]{{\setlength{\fboxsep}{0pt}\colorbox{#1}{\ifmmode$\displaystyle#2$\else#2\fi}}}
\newcommand{\gamechangem}[2][highlight-magenta]{{\setlength{\fboxsep}{0pt}\colorbox{#1}{\ifmmode$\displaystyle#2$\else#2\fi}}}

\newcommand{\gamechanger}[2][highlight-green]{{\setlength{\fboxsep}{0pt}\colorbox{#1}{\ifmmode$\displaystyle#2$\else#2\fi}}}
\colorlet{highlight-green}{green!40!yellow}

\newcommand{\Func}{\mathsf{Fcs}}
\newcommand{\Perm}{\mathsf{Perm}}

\newcommand{\ro}{\mathbf{H}}

\newcommand{\dist}{\mathsf{hide}}
\newcommand{\guess}{\mathsf{ow}}

\newcommand{\undist}{\mathsf{un\mbox{-}hide}}
\newcommand{\unguess}{\mathsf{un\mbox{-}ow}}

\newcommand{\pzdist}{\mathsf{pz\mbox{-}hide}}
\newcommand{\pzguess}{\mathsf{pz\mbox{-}ow}}

\newcommand{\abkhdist}{\mathsf{perm}}

\newcommand{\figref}[1]{Fig.~\ref{#1}}
\newcommand{\secref}[1]{Section~\ref{#1}}
\newcommand{\apref}[1]{Appendix~\ref{#1}}
\newcommand{\thref}[1]{Theorem~\ref{#1}}

\newcommand{\procRo}{\textsc{Ro}}
\newcommand{\procOrac}{\textsc{O}}

\newcommand{\procPerm}{\textsc{Perm}}

\newcommand{\procReprogram}{\textsc{Rep}}
\newcommand{\procFReprogram}{\textsc{FRep}}

\newcommand{\procFRo}{\textsc{FRo}}

\newcommand{\procSamp}{\textsc{Samp}}
\newcommand{\procFSamp}{\textsc{FSamp}}

\newcommand{\had}{\mathcal{H}}

\newcommand{\gamesfontsize}{\small}

\newcommand{\oneCol}[2]{
\begin{center}
        \framebox{
        \begin{tabular}{c}
        \begin{minipage}[t]{#1\textwidth}
        	\gamesfontsize #2 \end{minipage}
        \end{tabular}
        }
\end{center}
}

\newcommand{\oneColNoBox}[1]{
\begin{center}\begin{tabular}{c}
\begin{minipage}[t]{1in}\begin{tabbing}
123\=123\=123\=\kill
#1
\end{tabbing}\end{minipage} 
\end{tabular}\end{center}}

\newcommand{\twoColsNoBoxNoDivide}[2]{
\begin{center}\begin{tabular}{ccc}
\begin{minipage}[t]{1in}\begin{tabbing}
12\=12\=12\=\kill
#1
\end{tabbing}\end{minipage} & \hspace{10pt} &
\begin{minipage}[t]{1in}\begin{tabbing}
12\=12\=12\=\kill
#2
\end{tabbing}\end{minipage} 
\end{tabular}\end{center}}

\newcommand{\twoColsNoDivide}[4]{
\begin{center}
        \framebox{
        \begin{tabular}{c@{\hspace*{.4em}}c}
        \begin{minipage}[t]{#1\textwidth}
        	\gamesfontsize #3 \end{minipage}
        &
        \begin{minipage}[t]{#2\textwidth}
        	\gamesfontsize #4 \end{minipage}
        \end{tabular}
        }
\end{center}
}

\newcommand{\threeColsNoDivide}[6]{
\begin{center}
        \framebox{
        \begin{tabular}{c@{\hspace*{.4em}}c@{\hspace*{.4em}}c}
        \begin{minipage}[t]{#1\textwidth}
        	\gamesfontsize #4 \end{minipage}
        &
        \begin{minipage}[t]{#2\textwidth}
        	\gamesfontsize #5 \end{minipage}
        &
        \begin{minipage}[t]{#3\textwidth}
        	\gamesfontsize #6 \end{minipage}
        \end{tabular}
        }
\end{center}
}

%% file: abstract.tex

\begin{abstract}
	An important proof technique in the random oracle model involves reprogramming it on hard to predict inputs and arguing that an attacker cannot detect that this occurred.
	In the quantum setting, a particularly challenging version of this considers adaptive reprogramming wherein the points to be reprogrammed (or the output values they should be programmed to) are dependent on choices made by the adversary.
	Some quantum frameworks for analyzing adaptive reprogramming were given by Unruh (CRYPTO 2014, EUROCRYPT 2015), Grilo-H{\"o}velmanns-H{\"u}lsing-Majenz (ASIACRYPT 2021), and Pan-Zeng (PKC 2024). 
	We show, counterintuitively, that these adaptive results follow from the \emph{nonadaptive} one-way to hiding theorem of Ambainis-Hamburg-Unruh (CRYPTO 2019).
	These implications contradict beliefs (whether stated explicitly or implicitly) that some properties of the adaptive frameworks cannot be provided by the Ambainis-Hamburg-Unruh result.
\end{abstract}

%% file: intro.tex

\section{Introduction}



Hash functions are a pillar of modern practical cryptography.
Many of the most efficient algorithms for a given task (public key encryption, digital signatures, authenticated key exchange, \dots) crucially rely on hash functions for their security. 
Often the security cannot be reduced to standard model assumptions about the hash function and is instead justified by modeling it as a random oracle~\cite{CCS:BelRog93}.
A variety of methods of expressing random oracle model proofs are known which can make it easy to express analyses based on basic probabilistic analysis~\cite{EC:BelRog06,EC:CheSte14,EC:Maurer02,SAC:Patarin08,EPRINT:Shoup04}.
An important aspect of this is bounding the probability that an attacker notices if we adaptively modify the behavior of the oracle on inputs which are statistically/computationally hard for the attacker to guess.

As we prepare for a post-quantum future, there is a natural desire to port these benefits over to the quantum random oracle model~\cite{AC:BDFLSZ11} which captures that an attacker locally computing a hash function may use a quantum computer to do so in superposition. 
With this change in model, existing classical proofs no longer apply and, more importantly, the standard proof techniques no longer work.
Motivated by this challenge, many techniques have been introduced for quantum random oracle model analysis~\cite{C:AmbHamUnr19,TCC:BHHHP19,SAC:Eaton17,AC:GHHM21,PKC:HulRijSon16,C:JZCWM18,EC:KSSSS20,PKC:PanZen24,C:Unruh14,EC:Unruh15,unruh,FOCS:Zhandry12,C:Zhandry12,C:Zhandry19}.

Rather than providing new analysis tools, we show that an existing tool by Ambainis, Hamburg, and Unruh~\cite{C:AmbHamUnr19} (AHU) is more broadly applicable than previously realized.\footnote{In fact, AHU provide multiple O2H theorems. We focus on applying their Theorem 3, which is derived from the Semi-classical O2H theorem (their Theorem 1) which is the focus of their work.}
This ``one-way to hiding'' (O2H) theorem (building on earlier versions~\cite{SAC:Eaton17,C:JZCWM18,C:Unruh14,EC:Unruh15,unruh}) bounds an attacker's advantage in distinguishing between a pair of randomly chosen functions based on the probability a related algorithm can extract an input on which they differ. 
The common impression~\cite{C:AmbHamUnr19,AC:GHHM21,PKC:PanZen24} seems to be that this result is highly non-adaptive because these two functions must be fixed at the beginning of the experiment, and thus it cannot be used for adaptive analysis where the definitions of the functions can update throughout the experiment (except in a few special corner cases where an adaptive problem can cleverly be expressed nonadaptively).

We dispel this notion, showing that through a change of viewpoint we can easily analyze many highly adaptive problems.
As concrete applications of this we prove that an ``adaptive reprogramming framework'' of Pan and Zeng~\cite{PKC:PanZen24}, a ``tight adaptive reprogramming theorem'' of Grilo, H{\"o}velmanns, H{\"u}lsing, Majenz~\cite{AC:GHHM21}, and ``adaptive O2H'' lemmas of Unruh~\cite{C:Unruh14,EC:Unruh15} are all implied by the AHU O2H result.
We use straightforward proofs that largely rely on classical-like reasoning.
The proofs of the ``tight adaptive reprogramming theorem'' and one of Unruh's ``adaptive O2H'' lemmas start by representing a random oracle sparsely with Zhandry's technique~\cite{C:Zhandry19}.
Our concrete bounds are essentially equivalent to or better than the existing bounds.

Our approach uses a proof technique of Jaeger, Song, and Tessaro~\cite{TCC:JaeSonTes21} which applied a variant of AHU's O2H with fixed permutations to analyze the quantum security of a key-length extension technique they call FFX.
We show that this technique is broadly applicable for proving adaptive reprogramming results.
The main idea underlying this technique is to switch away from viewing the pair of O2H functions in AHU's theorem as strictly applied to the attacker's input. 
We instead use permutations that take a combined input containing both the attacker's input and the current state of the ``security game'', using the latter to respond to the former. 
Updating the state of the game thereby adaptively changes the inputs on which the oracles differ \emph{from the attacker's perspective}, but actually the set of \emph{combined} inputs on which the permutations differ is fixed from the beginning.
In essence, we view the O2H distinguisher as internally running both the attacker and the security game, only exporting very specific pieces of the computation to an oracle.
The idea of having the O2H distinguisher run both the attacker and the security game is the core technical idea that allows adaptive proofs using AHU's O2H.

\subsection{Technique Overview}
Consider a setting where an attacker $\advA$ makes $q$ queries $x_1,\dots,x_q$ to random oracle $H$.
It's allowed to adaptively select points $x_i^{\ast}$, asking for $H(x_i^{\ast})$ to be redefined to some $y_i^{\ast}$ up to $n$ times.
Let $p_1$ be the probability it outputs $1$ at the end of its execution and $p_0$ be the probability it outputs $1$ if instead $H$ is never redefined. 
Our goal is to bound $|p_1-p_0|$.

This setting is adaptive in two senses.
First, the differences between the oracles in the experiment change over time, rather than being chosen at the beginning of the experiment. 
Second, the particular ways in which the oracles are changed may depend on earlier queries of the adversary.

\heading{Classical Analysis}
To analyze this classically, we might use a Bellare-Rogaway style ``equivalent-until-bad'' approach~\cite{EC:BelRog06} (or analogous approaches by Maurer~\cite{EC:Maurer02} or Shoup~\cite{EPRINT:Shoup04}).
Therein we express this setting as a pair of pseudocode games, parameterized by a bit $b$, which are syntactically identical except after a flag $\mathsf{bad}$ is set. 
This could, for example, be captured by using the pseudocode oracle $O_b$ on the left of \figref{fig:bad}.
Now the fundamental lemma of game playing tells us that $|p_0-p_1|\leq\Pr[\mathsf{bad}]$.

At this point, we would bound $\Pr[\mathsf{bad}]$ based on some assumption about $\advA$ (e.g., that the $x_i^{\ast}$ are statistically or computationally unpredictable).
Applying a union bound across all of the queries of $\advA$ gives
\[
	\Pr[\mathsf{bad}] \leq \sum_j \Pr[x_j \in \sett{x_i^{\ast}}]= q\mathop{\mathbb{E}}_{j}[\Pr[x_j \in \sett{x_i^{\ast}}]].
\]
For example, if the $x_i^{\ast}$ are from adaptively chosen distributions that always have min-entropy at least $\mu$ and the $x_i^{\ast}$ are used nowhere else we get $|p_0-p_1|\leq nq2^{-\mu}$.
Notably, (treating the fundamental lemma of game playing as given) the analysis consists entirely of syntactic rewriting of the setting as pseudocode combined with basic probability calculation.

\begin{figure}[t]
\twoColsNoBoxNoDivide{
\underline{$O_b(x_j)$}\\
If $\exists i, x_i^\ast=x_j$: \\
\ind $\mathsf{bad}\gets \true$\\
\ind If $b=1$: Return $y_i^\ast$\\ 
Return $H(x_j)$\smallskip
}
{
\underline{$P_{b}(X, Y : H, X^{\ast}, Y^{\ast}, I)$}\\
If $\exists i \leq I, X_i^\ast = X$\comment{$\mathsf{bad}$}\\
\ind If $b=1$: $Y\gets Y \xor Y_i^\ast$\\
\ind If $b=0$: $Y\gets Y \xor H(X_i^{\ast})$\\
Else $Y\gets Y \xor H(X)$\\
Return $(X,Y:H,X^{\ast},Y^{\ast},I)$\smallskip
}
\caption{\textbf{Left:} Expression of oracle as pseudocode for classical equivalent-until-bad analysis. \textbf{Right:} Expression of oracle as a classical permutation queried in superposition for quantum Fixed-Permutation O2H analysis.}
\label{fig:bad}
\hrulefill
\end{figure}

\heading{Nonadaptive O2H}
Now consider a setting where $x_i^{\ast},y_i^{\ast}$ are still chosen classically, but the attacker has quantum access to $H$.
The standard formalization of this allows computing the classical permutation $H[\xor]:(x,y)\mapsto(x,y\xor H(x))$ in superposition.
We can no longer apply the approach above, as the oracle can be queried in a superposition over all $x$ at once so the bad event that it was queried on some $x_i^{\ast}$ is not even well defined.
One approach to this is ``one-way to hiding'' (O2H) results~\cite{C:AmbHamUnr19,SAC:Eaton17,C:JZCWM18,C:Unruh14,EC:Unruh15,unruh} which can be thought of as a quantum analog to equivalent-until-bad analysis.
Consider using an O2H lemma of AHU~\cite[Theorem 3]{C:AmbHamUnr19}, which (when simplified) considers a distribution over functions $H_0$, $H_1$ and tells us that 
\[
	|p_{H_0}-p_{H_1}|\leq 2q\sqrt{\mathop{\mathbb{E}}_{j}[\Pr[\textrm{Measure}(X_j) \in S]]}.
\]
Here $p_{H_b}$ is the probability an adversary outputs 1 when interacting with $H_b[\xor]$ and the last probability considers running the adversary $\advA$ with access to either oracle $H_b$, measuring its $j$-th query to the oracle, and checking whether the resulting measured value $x_j$ is in the set $S=\sett{x~:~H_0(x)\neq H_1(x)}$.

Unfortunately, this result is nonadaptive.
The two functions $H_0$ and $H_1$ (and thus the points where they differ) are fixed at the beginning of the game.
Consequently, it \emph{seems} that the result can only be applied in the limited case that all $x_i^{\ast},y_i^{\ast}$ are chosen at the beginning of the game and thus the programming of $H$ occurs immediately.
(For this, consider the distribution that samples the $x_i^{\ast},y_i^{\ast}$, defines $H_0$ to be a random oracle, and defines $H_1$ to equal $H_0$ except it is reprogrammed so $H_1(x_i^{\ast})=y_i^{\ast}$ for all $i$.)
For example, if the $x_i^{\ast}$ are from distributions with min-entropy at least $\mu$ and the $x_i^{\ast}$ are used nowhere else we get $|p_0-p_1|\leq \sqrt{n q^2 2^{-\mu}}$.

\heading{Fixed-Permutation O2H}
This impression is incorrect. 
We can analyze many adaptive reprogramming settings using AHU's O2H.
Start by simplifying it so that rather than considering distributions over permutations of the form $H_b[\xor]$ we consider just two permutations $P_0$ and $P_1$ that are a priori fixed, rather than chosen at random.
It follows from a Jaeger, Song, and Tessaro~\cite{TCC:JaeSonTes21} variant of AHU's O2H lemma that
\[
	|p_{P_0}-p_{P_1}|\leq 2q\sqrt{\mathop{\mathbb{E}}_{j}[\Pr[\textrm{Measure}(~P_0(X_j)\neq P_1(X_j)~)]]}.
\]
We call this the Fixed-Permutation O2H.
To apply this result, we switch our viewpoint.
Rather than thinking of $P_b$ being the functions the attacker $\advA$ might have access to (and so the O2H distinguisher is essentially identical to the attacker) we will think of building a distinguisher $\advD$ that jointly runs both $\advA$ and the ``security game'' $\advA$ is interacting with.
Then the permutations $P_b$ will be permutations that process oracle queries as a function of both the attacker's input and the game's state. 

In our running example, to simulate the game the distinguisher will store the random function in a quantum register $H$ and the reprogramming points $x_i^{\ast},y_i^{\ast}$ in registers $X^{\ast}_i$ and $Y^{\ast}_i$ (together with a register $I$ counting how many points have been chosen so far). 
When the original attacker $\advA$ wants to query $H_b[\xor]$ with registers $X,Y$ while being run internally by the distinguisher $\advD$, the distinguisher forwards this together with the game registers $H, X^{\ast}, Y^{\ast}, I$ as a query $(X,Y,H, X^{\ast}, Y^{\ast}, I)$ to its oracle $P_b$ defined on the right side of \figref{fig:bad}.
It returns just the $X$ and $Y$ registers to $\advA$.
Notationally, we use a colon as syntactic sugar to distinguish the input registers controlled by the attacker $\advA$ or the game it is interacting with.
We emphasize, however, that $P_b$ is a permutation on six inputs, \emph{all} of which are controlled by the distinguisher.

In this way, the distinguisher perfectly simulates the view of the attacker and the two permutations only differ on inputs for which $X$ is one of the reprogrammed points (as specified by $X^{\ast}$ and $I$).
Thus the above bound on $|p_{P_0}-p_{P_1}|$ gives a bound on $|p_0-p_1|$.
For example, if the $x_i^{\ast}$ are from adaptively chosen distributions that always have min-entropy at least $\mu$ and the $x_i^{\ast}$ are used nowhere else we get $|p_0-p_1|\leq \sqrt{n q^2 2^{-\mu}}$.
Notably, (treating the Fixed-Permutation O2H as given) the analysis consisted entirely of syntactic rewriting of the oracles as permutations, combined with basic probability calculation.

If we considered an arbitrary distinguisher $\advD$, the Fixed-Permutation O2H would be vacuous because $\advD$ could always just make all its queries with $I>1$ and $X=X^*_i$, making the right sight of the bound equal to one.
However, we will \emph{not} apply the bound with arbitrary distinguishers.
We apply it with distinguishers of the particular form described above for whom the value of the right side of the bound depends on $\advA$.

\heading{Backwards Bounds from Sparse Functions}
Thinking classically, if each $y_i^\ast$ is uniformly random, then swapping it with $H(x_i^{\ast})$ will only actually be detectable if the attacker queries its oracle at $x_i^{\ast}$ \emph{both} before and after the reprogramming.
This can be important because in some settings $x_i^{\ast}$ may be hard to query to $H$ before the reprogramming occurred but easy to query afterwards.
Our approach so far only considering the probability of $x_i^{\ast}$ being queried after reprogramming and thus could not be applied. 

\begin{figure}[t]
\twoColsNoBoxNoDivide{
\underline{$\textsc{CRo}(x)$}\\
If $T[x]=\bot$: Sample $T[x]$\\
Return $T[x]$\\[4pt]
\underline{$\textsc{CRep}_b(x_i^{\ast})$}\\
If $T[x_i^{\ast}]\neq\bot$:\\
\ind $\mathsf{bad}\gets\true$\\
\ind If $b=1$: Sample $T[x_i^{\ast}]$\\
Return
\smallskip
}
{
\underline{$\procFRo(X,Y : H)$}\\
$H[X] \gets H[X] \xor Y$\\
Return $(X,Y : H)$\\[4pt]
\underline{$\procFReprogram_b(X^{\ast} : H, I, Z)$ }\\
If $H[X^{\ast}] \neq Z[I]$:\comment{bad}\\
\ind If $b=1$: Swap $H[X^{\ast}]$ and $Z[I]$\\
$I\gets I+1\bmod n$\\ 
Return $(X^{\ast} : H, I, Z)$\smallskip
}
\caption{\textbf{Left:} Pseudocode for classical proof using lazily sampled random function. \textbf{Right:} Permutations for Fixed-Permutation O2H proof using sparse representation of random function. Tables $H$ and $Z$ are initially all zero.}
\label{fig:bad2}
\hrulefill
\end{figure}

Fortunately, the solution to this in the classical setting can be applied to the Fixed-Permutation O2H approach as well. 
The core idea is to consider the ``bad event'' as occurring when $x_i^{\ast}$ is chosen.
Then we will look backwards in time to see if $x_i^{\ast}$ was previously queried to the random oracle.
To do this, we lazily sample the random oracle using a table and then check whether a new $x_i^{\ast}$ matches any of the values currently in the table.
This is captured by the pseudocode oracles $\textsc{CRo}$ and $\textsc{CRep}_b$ on the left of \figref{fig:bad2} for querying the random oracle and reprogramming it to a random value on a specified input.
Applying the fundamental lemma of game playing and union bounds gives. 
\[
	|p_0-p_1| 
	\leq \sum_i \Pr[x_i^{\ast} \in T]= 
	n\mathop{\mathbb{E}}_{i}[\Pr[x_i^{\ast} \in T]].
\]
If the $x_i^{\ast}$ are from adaptively chosen distributions that always have min-entropy at least $\mu$ we get $|p_0-p_1|\leq nq2^{-\mu}$.
Notably, (treating the fundamental lemma of game playing as given) the analysis consists entirely of straightforward syntactic rewriting of the setting as pseudocode combined with basic probability calculation.

For porting this idea over to the quantum regime, the important aspect of lazy sampling was that the random oracle was represented by a sparse table (one which had been written to in at most $q$ locations after $q$ queries).
Zhandry~\cite{C:Zhandry19} showed that we can similarly represent quantum random oracles with (superpositions over) sparse tables.
Thereby, we can move from the traditional quantum random oracle which does $y \gets y \xor H(x)$ for an $H$ initialized at random to the Fourier random oracle which does $H(x) \gets y \xor H(x)$ for an $H$ initialized to be all zero.
In this domain, reprogramming $H(x^{\ast}_i)$ to a random output can be performed by swapping its register with a fresh register initialized to be zero.
This is captured by the permutations $\textsc{FRo}$ and $\textsc{FRep}_b$ on the right of \figref{fig:bad2}.
Here registers $H$, $I$, and $Z$ are initialized as 0.

We can now apply the Fixed-Permutation O2H.
The distinguisher runs the attacker, internally storing the game registers and simulating $\procFRo$ for the attacker.
Whenever the original attacker $\advA$ wants to reprogram on the value in register $X^{\ast}$, the O2H distinguisher $\advD$ adds the game registers and then forwards this as a query $(X^{\ast},H,I,Z)$ to its own oracle $\procFReprogram_b$, then returns $X^{\ast}$ to the attacker.
Note that the $\procFReprogram_b$ permutations differ only if $H[X^{\ast}]$ and $Z[I]$ differ (and $Z[I]$ is necessarily zero), so we get
\[
	|p_0-p_1|\leq 2n\sqrt{\mathop{\mathbb{E}}_{i}[\Pr[\textrm{Measure}(H[X^{\ast}_i] \neq 0)]]}.
\]
This can give better bounds in the natural setting that $q\gg n$ because it (implicitly) switches a factor of $q\sqrt{n}$ with $n\sqrt{q}$.
For example, if the $x_i^{\ast}$ are from adaptively chosen distributions that always have min-entropy at least $\mu$ we get $|p_0-p_1|\leq \sqrt{n^2 q 2^{-\mu}}$.
Notably, (treating the Fixed-Permutation O2H as given) the analysis consisted entirely of straightforward syntactic rewriting of the oracles as permutations combined with basic probability calculation.\footnote{Prior to Zhandry's work, rewriting quantum random functions as sparse tables was not known to be straightforward. In light of the work, it is quite simple to do so for our purposes.}

\subsection{Applications of The Technique}
To show the broad applicability of the Fixed-Permutation O2H~\cite{TCC:JaeSonTes21} for adaptive proofs, we use it to imply an ``adaptive reprogramming framework'' of Pan and Zeng~\cite{PKC:PanZen24}, a ``tight adaptive reprogramming theorem'' of Grilo, H{\"o}velmanns, H{\"u}lsing, Majenz~\cite{AC:GHHM21} (GHHM), and ``adaptive O2H'' lemmas of Unruh~\cite{C:Unruh14,EC:Unruh15}.
We summarize these below.

For some of these results, the concrete bounds we establish have to be parameterized slightly differently than the original results, but we show that they are essentially equivalent (or better) in actual use.
Moreover, AHU's O2H shows that their upper bound also applies to the difference of square roots $|\sqrt{p_0}-\sqrt{p_1}|$.
This gives better bounds when used to prove $p_0$ is small based on a $p_1$ which is known to be small.
For example, if $p_1 = 0$ and the upper bound is $\epsilon$ we get $p_0 = (\sqrt{p_0} - \sqrt{p_1} )^2 \leq \epsilon^2$ which improves over $p_0 = p_0 - p_1 \leq \epsilon$.
We get the square root versions of all the results we consider for free.\footnote{As pointed out by an ASIACRYPT reviewer, a bound of $|p_0-p_1|\leq \epsilon$, generically implies a square root bound via $\sqrt{p_0}-\sqrt{p_1} \leq \sqrt{p_1+\epsilon} - \sqrt{p_1} \leq \sqrt{p_1} + \sqrt{\epsilon} - \sqrt{p_1} \leq \sqrt{\epsilon}$. Our bounds improve on this by getting $\sqrt{p_0}-\sqrt{p_1} \leq \epsilon$.}

\heading{Pan-Zeng Adaptive Reprogramming Framework}
Pan and Zeng~\cite{PKC:PanZen24} introduced an adaptive reprogramming framework which they use to analyze the selective-opening security of Fujisaki-Okamoto-style public key encryption algorithms.
They express a belief that AHU's O2H result lacked the properties needed for these proofs, saying that,
\begin{quote}
	Our core technical contribution is a computational adaptive reprogramming framework in the QROM that enables a security reduction to adaptively and simultaneously reprogram polynomially many RO-queries which are computationally hidden from a quantum adversary.
	This is a property that \emph{cannot be provided} by previous techniques in the QROM, such as \dots the semi-classical O2H lemma~\cite{C:AmbHamUnr19}\dots\footnote{This quote is from p.4 (aka p.95) of the proceedings version or p.3 of the current ePrint version. Emphasis ours. We've changed the citation to match our numbers.}
	\end{quote}
We prove that their computational adaptive reprogramming result is implied by the Fixed-Permutation O2H with a short proof, thereby establishing that the O2H lemma \emph{can} provide this property.\footnote{Technically, the mentioned semi-classical O2H lemma is a different O2H result of AHU, but AHU showed~\cite{C:AmbHamUnr19} it directly implies the O2H lemma that we prove equivalent to the Fixed-Permutation O2H in \secref{sec:prelims}.}
Their framework considers arbitrary reprogramming of the oracle and upper bounds distinguishing advantage by the probability that measuring the input of a random query gives a value at which the function differ.
In essence, their result follows directly from the first use of Fixed-Permutation O2H that we described above.

\heading{GHHM Adaptive Reprogramming Framework}
GHHM~\cite{AC:GHHM21} gave a tight adaptive reprogramming theorem for information theoretic settings where the reprogrammed points are from adaptively chosen distributions with high min-entropy, but are immediately given to the attacker. 
Consequently, the distinguishing advantage must be bound by the probability that one of the reprogrammed points is queried before being selected and does not seem to imply or be implied by the Pan-Zeng result.\footnote{The paper~\cite{PKC:PanZen24} erroneously implies that the GHHM theorem is captured by the Pan-Zeng framework. We pointed this error out to the authors who agreed that the current result does not imply the GHHM theorem.} 
Our proof of GHHM's theorem follows from our second use of Fixed-Permutation O2H by using Zhandry's technique for sparsely representing functions to provide ``backwards bounds''.
GHHM used their tool for tighter proofs of hash and sign techniques (e.g., used by XMSS), tighter proofs for Fiat Shamir signatures, and fault resistance for the hedged Fiat-Shamir transform.
Their theorem was later used by~\cite{EC:ABKM22,C:DFPS23,C:DFHS24,PKC:KosXag24,AC:MorYam22,ACISP:YuaSunTak24,ACISP:YuaTibAbe23}.

\heading{Unruh's Adaptive O2H}
Unruh~\cite{C:Unruh14,EC:Unruh15} gave adaptive variants of early O2H results for reprogramming on a single statistically hidden input.
These results obtain improved concrete bounds by separately considering the probability that reprogrammed point is queried before or after it is sampled.
Consequently, they do not seem to imply or be implied by either the Pan-Zeng or the GHHM result.

Notably, AHU proved a theorem (their Theorem 4) that had previously been shown using the first adaptive O2H result~\cite{C:Unruh14}.
While doing so, they note that ``at least in the proof from~\cite{TCC:TarUnr16}'' they could replace the adaptive O2H result with their nonadaptive version by programming the random oracle on many points.
Our proof of the~\cite{C:Unruh14} result applies the Fixed-Permutation O2H two times, one of which similarly programs the random oracle on many points.
Thereby we show that the approach of AHU actually extends to any application of the adaptive O2H, not just that particular proof. 

We reprove the~\cite{EC:Unruh15} result with two applications of the Fixed-Permutation O2H using ``backwards bounds''.
This setting is more general than in~\cite{C:Unruh14}, but the concrete bounds based on collision-entropy of the input distribution are incomparable.
Our bound is in terms of min-entropy which implies both the collision-entropy bound~\cite{EC:Unruh15} and an improved version of the bound in~\cite{C:Unruh14} (replacing a $q_0$ factor with $\sqrt{q_0}$). 

\heading{A Non-Application}
The Fixed-Permutation O2H is not a panacea.
We conclude the paper by discussing two results that seem out of reach of the Fixed-Permutation O2H with current techniques.
The first result, by Alagic, Bai, Katz, and Majenz~\cite{EC:ABKM22}, gives a variant of the GHHM result for random permutations reprogrammed on uniformly random points.
The result proves an $O(\sqrt{q/2^n})$ bound while we are only able to prove $O(\sqrt{q^2/2^n})$.
The second result, by Alagic, Bai, Katz, Majenz, and Struck~\cite{EC:ABKMS24}, generalizes the result when a reprogrammed point comes from an adaptively chosen, high entropy distribution.
For both, we identify that the gap between our success in proving the GHHM result and inability to prove these results stems from lacking techniques for expressing quantum random permutations sparsely.

\subsection{Implications of the Results}
Formally, the claim that Fixed-Permutation O2H implies any of these other theorems is essentially tautological.
The results were already unconditionally proven to be true, so it is vacuously the case that any statement implies them.
The essence of the result is not that the implications hold, but rather that the proofs thereof are straightforward and require almost exclusively classical reasoning.
(The quantum complexity is instead hidden inside of the O2H result we take as assumed.)

There are two ways we imagine this being used in future work. 
If one likes to have a ``toolbox'' of adaptive reprogramming results each targeted narrowly at a particular type of problem that they are well suited to expressing, then our results show that the Fixed-Permutation O2H is useful to build such tools.
Alternatively, because of the simplicity of our proofs, one could choose to jettison the use of individual adaptive reprogramming results and instead use Fixed-Permutation O2H directly in security proofs as a single powerful ``multi-tool''.\footnote{As an example, we note that several works~\cite{C:DFHS24,AC:GHHM21,PKC:KosXag24,ACISP:YuaSunTak24} use both GHHM's result and O2H theorems from AHU.}

In subsequent work, Liao, Ge, and Xue~\cite{EPRINT:LiaJiaXue25} introduce a Double-Oracle Fixed-Permutation O2H and show it can replace the Fixed-Permutation O2H in some proofs.
This improves the bound in the GHHM theorem and matches the bound in the Alagic, Bai, Katz, and Majenz theorem.

%% file: prelims.tex

\section{Preliminaries}\label{sec:prelims}

\heading{Notation}
We write $y\getsr\advA[O](x)$ for randomized execution of $\advA$ with input $x$ and oracle access to $O$ which produces output $y$.
We consider quantum $\advA$ that can access $O$ in superposition.
We let $\Pr[\sG]$ denote the probability that game $\sG$ returns $\true$.
Registers are implicitly initialized to store the all zero string.

If $\mathcal{S}$ is a set, then $y\getsr\mathcal{S}$ denotes sampling $y$ uniformly from $\mathcal{S}$.
We let $\Func(n,m)$ denote the set of all functions mapping $\bits^n$ to $\bits^m$.
Sampling $H \getsr \Func(n,m)$ gives a uniform random function.

\subsection{Quantum Computation Background}\label{sec:quant-back}

We assume familiarity with basic quantum computation, as performing unitary operations on registers which each contain a fixed number of qubits that can be measured in the computational basis.
Our main results are primarily based on a ``one-way to hiding'' theorem (defined soon) which when treated as a blackbox allow us to primarily think ``classically''.
We summarize the most important ideas used in our proofs.

\heading{Computing Permutations}
If $P$ is a permutation, then there is a quantumly computable unitary $U_P$ which maps according to $U_P\ket{x} = \ket{P(x)}$ for $x\in\bits^n$.
The runtime of this unitary grows with the time required to compute both $P$ and $P^{-1}$ classically.
We write $P$ in place of $U_P$.
If $f:\bits^n\to\bits^m$ is a function, we define the permutation $f[{\xor}]$ by $f[{\xor}](x,y)=(x,f(x)\xor y)$.

We define permutations $\procOrac$ that will be provided as (quantum accessible) oracles using the notation:
\oneColNoBox{
\underline{$\procOrac(X_1,\dots, X_m : Z_1,\dots, Z_n)$}\\
//Code updating $X_1,\dots,X_m$ and $Z_1,\dots,Z_n$\\
Return $(X_1,\dots,X_m  : Z_1,\dots,Z_n)$\smallskip
}
Writing $\advA[O]$ will denote the variables/registers before the colon ($X_1,\dots,X_m$) are controlled by $\advA$ while variables/registers after the colon ($Z_1,\dots,Z_n$) are controlled by its surrounding security game.
Writing $\advD[O_{\Leftrightarrow}]$ will denote that $\advD$ controls all of the variables/registers.
We use this for one-way-to-hiding distinguishers (which when used for reductions will internally run code of \emph{both} the attacker and the game the attacker expects to interact with).

\heading{Principle of Deferred Measurement}
In proofs, we find it convenient to defer any classical measurements until the end of execution by writing the result of the measurement (in superposition) into an auxiliary register that will otherwise be unused.
If at the end of execution we only care about the measured value of some other register, then the measurement can be deferred indefinitely.

\heading{Sparse Representation of a Uniformly Random Function}
In two proofs we make use of Zhandry's~\cite{C:Zhandry19} technique for representing random functions with sparse tables.
It lets us replace access to an oracle for which $(x,y)$ is mapped to $(x, y \xor T[x])$ for a uniformly random $T$ with an oracle for which $(x,y)$ is unchanged and $T[x]$ is updated by $T[x] \gets T[x] \xor y$ for an initially all-zero $T$.

The full compressed oracle technique of Zhandry combines the above with the ability to represent such a sparse table compactly.
We do not need this compactness.
When we first make use of this approach we will, in an appendix, provide a rigorous breakdown of the technique into individual steps for readers unfamiliar with the technique.

Unruh~\cite{AC:Unruh23} gives a generalization of Zhandry's technique, exhibiting that the particular choices of using the Fourier domain or the Hadamard transform are inessential for the result.
However, for our analysis it is more convenient to concretely work with the Hadamard transform because it allows us to concretely write simple pseudocode for the sparse random oracle. 
Unruh's presentation does not seem to easily allow this.

\ifnum\eprintversion=0
        \newcommand{\lwid}{0.4}
        \newcommand{\mwid}{}
        \newcommand{\rwid}{0.45}
\else
        \newcommand{\lwid}{0.16}
        \newcommand{\mwid}{}
        \newcommand{\rwid}{0.29}
\fi

\heading{(Fixed-Permutation) One-way to Hiding}
Ambainis, Hamburg, and Unruh (AHU)~\cite{C:AmbHamUnr19} proved a ``one-way to hiding'' (O2H) theorem which bounds the ability of a distinguisher to distinguish between two oracles by the probability that the distinguisher can be used to find an input on which the two oracles differ.
Their theorem and typical uses thereof consider distributions over the oracles.
We focus on using a variant where the oracles are permutations that are fixed ahead of time as first used by Jaeger, Song, and Tessaro~\cite{TCC:JaeSonTes21}.
\begin{wrapfigure}{r}{0.33\textwidth}\small
\oneCol{0.29}
{
\underline{Game $\sG^{\dist}_{O}(\advD)$}\\
$b\getsr\advD[O_{\Leftrightarrow}]$\\
Return $b=1$\\[4pt]
\underline{Game $\sG^{\guess}_{P,P'}(\advD)$}\\
$i\getsr\sett{1,\dots,q}$\\
Run $\advD[P_{\Leftrightarrow}]$ until its $i$-th query\\
Measure the input $x$ to this query\\
Return $(P(x)\neq P'(x))$\smallskip
}
\caption{Games used for O2H theorems}
\label{fig:reg-ahu}
\end{wrapfigure}
We will define both and show that they are essentially equivalent. 

Consider the game $\sG^{\dist}$ defined in \figref{fig:reg-ahu} wherein the distinguisher $\advD$ is given access to an oracle $O$ and then outputs a bit $b$.
For $e\in\sett{1,1/2}$, we measure the ability of $\advD$ to distinguish between permutations $P$ and $P'$ by
\[\Adv^{\dist}_{P,P',e}(\advD)=\left(\Pr[\sG^{\dist}_{P}(\advD)]\right)^e-\left(\Pr[\sG^{\dist}_{P'}(\advD)]\right)^e.\]

The O2H theorem bounds this in terms of the game $\sG^{\guess}$ shown in the same figure.
There the distinguisher is run with access to oracle $P$.
One of its oracle queries (chosen at random) is measured and the game returns $\true$ is the permutations $P$ and $P'$ would give different outputs on this input.
We define \[\Adv^{\guess}_{P,P'}(\advD)=\Pr[\sG^{\guess}_{P,P'}(\advD)].\]

\begin{theorem}[Fixed-Permutation O2H, \cite{TCC:JaeSonTes21}]\label{thm:reg-ahu}\label{o2h}
	Let $P,P'$ be permutations, $\advD$ be an distinguisher making at most $q$ oracle queries, and $e\in\sett{1,1/2}$.
	Then \[\abs{\Adv^{\dist}_{P,P',e}(\advD)}\leq 2q\sqrt{\Adv^{\guess}_{P,P'}(\advD)}.\]
\end{theorem}
Note that the two permutations are a priori fixed.
Assuming $P\neq P'$, there trivially exist distinguishers which can distinguish between the two permutations by simply querying them on an input where they differ.
Thus, when making productive use of this theorem we will always be considering some restricted class of distinguishers.
Generally, the distinguisher will internally be running both an adversary and a security game the adversary is interacting with. 
The permutations will be used to process when the adversary makes an oracle query to its game, with the distinguisher's query to the permutation consisting both of query variables chosen by the attacker and internal state variables of the security game.

The original AHU result considered distributions over oracles of the form $f[\xor]$, rather than arbitrary permutations.
Let $\distD$ be a distribution over $(f,f',\advD)$ where $f,f'$ are functions and $\advD$ is a distinguisher (for comparison to AHU's original statements, think of it as a fixed distinguisher on input a random string $z$).\footnote{AHU's original statement had $\distD$ pick a set $S$ satisfying $f'(x)=f(x)$ for all $x\not\in S$ and defined $\sG^{\guess}$ to check whether $x\in S$. Our formalism fixes the special case $S=\{x~:~f'(x) \neq f(x)\}$. As this is the smallest possible choice of $S$, it only makes $\sG^{\guess}$ harder and hence implies the results for general $S$.}
Then we define the advantage functions as follows 
where the expectations are over $(f,f',\advD)\getsr\distD$,
\begin{align*}
	\Adv^{\dist}_{e}(\distD)&=\left(\mathop{\mathbb{E}}_{}\left[\Pr[\sG^{\dist}_{f[\xor]}(\advD)]\right]\right)^e-\left(\mathop{\mathbb{E}}_{}\left[\Pr[\sG^{\dist}_{f'[\xor]}(\advD)]\right]\right)^e \text{ and }\\
	\Adv^{\guess}_{}(\distD)&=\mathop{\mathbb{E}}_{}\left[\Pr[\sG^{\guess}_{f[\xor],f'[\xor]}(\advD)]\right].
\end{align*}

We capture their theorem as follows, then show it is essentially equivalent to \thref{thm:reg-ahu}.
\begin{theorem}[\cite{C:AmbHamUnr19}, Thm.~3]\label{thm:reg-ahu-orig}
	Let $\distD$ be a distribution as above where $\advD$ makes at most $q$ oracle queries.
	Let $e\in\sett{1,1/2}$.
	Then
	\[
		\abs{\Adv^{\dist}_{e}(\distD)}\leq 2q\sqrt{\Adv^{\guess}_{}(\distD)}.
	\]
\end{theorem}

\begin{proposition}
	\thref{thm:reg-ahu} and \thref{thm:reg-ahu-orig} directly imply each other (up to constant factors).
\end{proposition}
\begin{proof}
	Let $\distD$ be given.
	Then define the permutations 
		\begin{align*}
			P(X,Y:f,f') &= (X,f(X)\xor Y:f,f')\\
			P'(X,Y:f,f') &= (X,f'(X)\xor Y:f,f').
		\end{align*}
	We views these as permutations applied to four inputs (three of which are output unmodified by the permutation).
	Then let $\advD^{\ast}$ sample $(f,f',\advD)\gets\distD$ and start running $\advD$ internally.
	Whenever $\advD$ makes an oracle query with registers $(X,Y)$, the distinguisher $\advD^{\ast}$ will query its oracle with $(X,Y:f,f')$.
	When $\advD$ halts and outputs $b$, $\advD^{\ast}$ halts and outputs $b$ as well.
	It is clear that
		\begin{align}
			\Pr[\sG^{\dist}_{P}(\advD^{\ast})] &= \mathop{\mathbb{E}}\left[\Pr[\sG^{\dist}_{f[\xor]}(\advD)]\right],	
			\Pr[\sG^{\dist}_{P'}(\advD^{\ast})] = \mathop{\mathbb{E}}\left[\Pr[\sG^{\dist}_{f'[\xor]}(\advD)]\right], \text{ and } \nonumber \\
			&\Pr[\sG^{\guess}_{P,P'}(\advD^{\ast})] = \mathop{\mathbb{E}}\left[\Pr[\sG^{\guess}_{f[\xor],f'[\xor]}(\advD)]\right]\label{eq:one}
		\end{align}
	and $\advD^{\ast}$ makes $q$ queries. Hence \thref{thm:reg-ahu} implies \thref{thm:reg-ahu-orig}.
	
	Now let $P$, $P'$, and $\advD^{\ast}$ be given where $\advD^{\ast}$ makes $q$ oracle queries.
	Define $\distD$ to be the distribution which always outputs $(P_{\pm},P'_{\pm},\advD)$ where $P_{\pm}(d,X) = P(X)$ if $d=1$ and $P^{-1}(X)$ if $d=0$.
	Permutation $P'_{\pm}$ is defined likewise.
	Distinguisher $\advD$ runs $\advD^{\ast}$ internally.
	Whenever $\advD^{\ast}$ makes an oracle query to its oracle with register $X$, $\advD$ prepares a register $Y=0^{|X|}$, queries $O[\oplus](1,X,Y)$ and $O[\oplus](0,Y,X)$ then swaps $X$ and $Y$ before returning $X$ to $\advD^{\ast}$.
	If the permutation underlying $O$ was $F$ and $(X,Y)$ initially contained $(x,0^{|x|})$, then the two oracle queries result in them containing $(0^{|x|},F(x))$, and so the final swap result is $(F(x),0^{|x|})$. 
	Because the second register began an ended as containing all $0$'s, this is indistinguishable from an oracle that just applied $F$ to $X$ from $\advD$'s perspective.
	When $\advD^{\ast}$ halts and outputs $b$, $\advD$ halts and outputs $b$ as well.
	
	The equalities in Eq.~\ref{eq:one} hold again, but now $\advD$ makes $2q$ oracle queries.
	Hence \thref{thm:reg-ahu-orig} implies \thref{thm:reg-ahu} up to an additional multiplicative factor of 2 being added to the latter theorem's bound.\qed
\end{proof}
As claimed by Jaeger, Song, and Tessaro~\cite{TCC:JaeSonTes21}, a direct emulation of the original proof for \thref{thm:reg-ahu-orig} in~\cite{C:AmbHamUnr19} gives the constant claimed in \thref{thm:reg-ahu}.

\heading{Parallel Query Bounds}
The original result of AHU is stronger than what is stated above because it provides better bounds against attackers that make queries in parallel. 
If the distinguisher makes up to $p$ parallel queries at a time over $d$ rounds, the above theorems would give us bounds of the form $2pd\sqrt{\Adv^{\guess}}$.
Instead, AHU show that its advantage grows with $2d\sqrt{\Adv^{\guess\mbox{-}\mathsf{par}}}$ where this advantage considers picking $i \getsr \sett{1,\dots,d}$, measuring all of the distinguisher's queries in the $i$-th round, and checking if any of them satisfy $(P(x)\neq P'(x))$.

We omit considering parallel queries for simplicity, but this form of bound will similarly hold for the Fixed-Permutation O2H.
It can naturally be extended to give parallel-query versions of our results in \secref{sec:panzen}, \secref{sec:unruh-one}, and \secref{sec:abkh}.
It would not naturally give meaningful parallel-query versions of our results in \secref{sec:ghhm} or \secref{sec:unruh2}.
The issue with these is that they apply the ``backwards bounds'' technique to provide improved bounds.
However, this then results in applying the Fixed-Permutation O2H to programming queries rather than random oracle queries.
Consequently, we would only get improved bounds for programming queries made in parallel, which does not seem as useful.

%% file: panzeng.tex

\section{Pan-Zeng Adaptive Reprogramming Framework}\label{sec:panzen}
In this section we prove that (a version of) the Pan-Zeng framework for computational adaptive reprogramming~\cite{PKC:PanZen24} is directly implied by the Fixed-Permutation O2H (\thref{o2h}).
We start by recalling their framework in \secref{sec:pz-frame}.
In \secref{sec:pz-implied}, we state and prove our variant of the framework.
The bounds provided by the results are complex and hard to compare. 
In \secref{sec:pz-comparisons} we apply the same simplifications that Pan and Zeng use when applying their result in theorems and show that our theorem provides better concrete bounds in this case.

\subsection{Pan-Zeng Framework and Security Theorem}\label{sec:pz-frame}
In the Pan-Zeng framework for computational adaptive reprogramming we consider multi-stage adversary $\advA=(\advA_0,\dots,\advA_n)$ trying to distinguish a nonadaptive world from an adaptive world.
This world is parameterized by an environment $\env$ which specifies $\mathsf{Init}$, $\mathsf{Orac}$, and $\mathsf{Repro}$.
The interactions are defined by the game $\sG^{\pzdist}$ defined in \figref{fig:pz-games}.

\ifnum\eprintversion=0
        \renewcommand{\lwid}{0.4}
        \renewcommand{\mwid}{}
        \renewcommand{\rwid}{0.45}
\else
        \renewcommand{\lwid}{0.23}
        \renewcommand{\mwid}{}
        \renewcommand{\rwid}{0.33}
\fi
\begin{figure}[t]
\twoColsNoDivide{\lwid}{\rwid}
{
\underline{Game $\sG^{\pzdist}_{\env,b}(\advA)$}\\
$(\mathsf{Init},\mathsf{Orac},\mathsf{Repro})\gets\env$\\
$(s,x,H_0,H_1)\getsr\mathsf{Init}$\\
$y \getsr \advA_0[H_b[\xor]](x)$\\
For $i=1,\dots,n$ do\\
\ind $(x,a)\getsr \mathsf{Orac}(s,y)$\\
\ind $H_1 \gets \mathsf{Repro}(s,a,H_1)$\\
\ind $y \getsr \advA_i[H_b[\xor]](x)$\\
Return $y=1$\smallskip
}
{
\underline{Game $\sG^{\pzguess}_{\env,i}(\advA)$}\\
$t\getsr\sett{1,\dots,q_i}$\\
Run $\sG^{\pzdist}_{\env,1}$ until $\advA_i$ is initiated\\
Run $\advA_i[H_1[\xor]](x)$ until its $t$-th query\\
Measure the input $X$ to this query\\
Return $(H_0(X)\neq H_1(X))$\smallskip
}
\caption{Games used for the Pan-Zeng computational adaptive reprogramming framework. Different stages of $\advA$ implicitly share state.}
\label{fig:pz-games}
\hrulefill
\end{figure}

First $\mathsf{Init}$ samples parameter string $s$ (later given to $\mathsf{Orac}$ and $\mathsf{Repro}$ to make them act in a coordinated manner), an initial input $x$ for $\advA$, and two functions $H_0$ and $H_1$.
In the nonadaptive world ($b=0$) each stage of $\advA$ is given oracle access to $H_0[\oplus]$, produces outputs $y$, and is given inputs $x$ from $\mathsf{Orac}$.
In the adaptive world ($b=1$) it is instead given access to $H_1[\oplus]$ and this function $H_1$ is adaptively updated by $\mathsf{Repro}$ in between each stage of $\advA$ based on an auxiliary string passed to it by $\mathsf{Orac}$.
In particular, $\advA$ passing $y$ to $\mathsf{Orac}$ represents $\advA$ having some sort of influence over how $H_1$ will be reprogrammed and $\mathsf{Orac}$ passing $x$ back represents some sort of leakage that $\advA$ learns about how its influence work.
The auxiliary information $a$ passed from $\mathsf{Orac}$ to $\mathsf{Repro}$ contains ``instructions'' on specifically how $\mathsf{Repro}$ should reprogram $H_1$.
The concrete details of how $\mathsf{Init}$, $\mathsf{Orac}$, and $\mathsf{Repro}$ will depend specifically on what the Pan-Zeng framework is being used to analyze.
However, we are analyzing the framework as a whole and so can be agnostic to the specifics of what they do. 

Note that the only quantum behavior of this game is internal computation by $\advA$ and its superposition queries to its $H_b[\xor]$ oracle.
For $e\in\sett{1,1/2}$, we define $\Adv^{\pzdist}_{\env,e}(\advA)
=
\left(\Pr[\sG^{\pzdist}_{\env,1}(\advA)]\right)^e
-
\left(\Pr[\sG^{\pzdist}_{\env,0}(\advA)]\right)^e.$

The Pan-Zeng framework (similar to O2H results) bounds the distinguishing advantage of $\advA$ in relation to an experiment where one of $\advA$'s queries are measured at random and we see if that query differentiates the two oracles.
Let $i\in\{0,\dots,n\}$ and define $q_i$ to be the number of oracle queries that $\advA_i$ makes.
This is captured by the game $\sG^{\pzguess}$ which is parameterized by $\env$ and the choice of $i$.
In it, we run $\advA$ in the adaptive world until stage $i$.
A random one of its queries in that stage are measured to see if $H_0$ and $H_1$ differ on that input.
We define $\Adv^{\pzguess}_{\env,i}(\advA)=\Pr[\sG^{\pzguess}_{\env,i}(\advA)]$.

Pan and Zeng proved the following (Lemma 2, proceedings version or Lemma 3.1, \href{https://eprint.iacr.org/2023/1682}{ePrint version}).
\begin{theorem}[Pan-Zeng Adaptive Reprogramming, \cite{PKC:PanZen24}]\label{thm:pz}
	Let $\env$ be an environment and $\advA$ be an adversary for which $\advA_i$ makes at most $q_i$ oracle queries.
	Then
	\[
		\abs{\Adv^{\pzdist}_{\env,1}(\advA)}
		\leq
		\sum_{k=0}^n \sum_{i=0}^k 2 q_i \sqrt{\Adv^{\pzguess}_{\env,i}(\advA)}.	\]
\end{theorem}

\subsection{The Pan-Zeng Theorem is Implied by O2H}\label{sec:pz-implied}
Now we state and prove a variant of \thref{thm:pz} which follows from the Fixed-Permutation O2H (\thref{o2h}).
\begin{theorem}\label{thm:pz-new}
	Let $\env$ be an environment and $\advA$ be an adversary for which $\advA_i$ makes at most $q_i\geq 1$ oracle queries.
	Let $q=q_0+\dots+q_n$.
	Then 
	for $e\in\sett{1,1/2}$
	it holds that
	\[
		\abs{\Adv^{\pzdist}_{\env,e}(\advA)}
		\leq
		2 q \sqrt{\sum_{i=0}^n \frac{q_i}{q} \Adv^{\pzguess}_{\env,i}(\advA)}.
	\]
\end{theorem}

\begin{proof}
	To apply the Fixed-Permutation O2H (\thref{o2h}) we will define appropriate $P$, $P'$, and $\advD$ from $\env$ and $\advA$.
	We define the permutations as follows.
	\begin{align*}
		P(X,Y : H_0, H_1) &= (X, Y \xor H_1(X) : H_0, H_1)\\
		P'(X,Y : H_0, H_1) &= (X, Y \xor H_0(X) : H_0, H_1)
	\end{align*}
	We views these as permutations with four inputs (three of which are returned unchanged).
	Note that $P(X,Y : H_0, H_1) \neq P'(X,Y : H_0, H_1)$ if and only if $H_1(X)\neq H_0(X)$.
	
	Now our distinguisher $\distD$ for $P$ and $P'$ will simply run $\sG^{\pzdist}_{\env,b}(\advA)$ (that is, internally running both $\advA$ and the algorithms of $\env$) except whenever $\advA$ would query $(X,Y)$ to $H_b[\xor]$ it will query $(X,Y : H_0, H_1)$ to its own oracle then return the resulting $(X,Y)$ to $\advA$.
	When $\advA$ produces its final output $y$, $\advD$ halts and outputs that as well.
	Formal pseudocode for $\advD$ is given in \figref{fig:pz-proof}.
	
	\ifnum\eprintversion=0
        \renewcommand{\lwid}{0.4}
        \renewcommand{\mwid}{}
        \renewcommand{\rwid}{0.45}
	\else
        \renewcommand{\lwid}{0.22}
        \renewcommand{\mwid}{}
        \renewcommand{\rwid}{0.31}
	\fi
	\begin{figure}[t]
	\twoColsNoDivide{\lwid}{\rwid}
	{
	\underline{Distinguisher $\advD[O_{\Leftrightarrow}]$}\\
	$(\mathsf{Init},\mathsf{Orac},\mathsf{Repro})\gets\env$\\
	$(s,x,H_0,H_1)\getsr\mathsf{Init}$\\
	$y \getsr \advA_0[S](x)$\\
	For $i=1,\dots,n$ do\\
	\ind $(x,a)\getsr \mathsf{Orac}(s,y)$\\
	\ind $H_1 \gets \mathsf{Repro}(s,a,H_1)$\\
	\ind $y \getsr \advA_i[S](x)$\\
	Return $y$\smallskip
	}
	{
	\underline{Oracle $S(X,Y)$}\\
	Call $O_{\Leftrightarrow}$ with input $(X,Y : H_0, H_1)$\\
	Return $(X,Y)$ to $\advA_i$\smallskip
	}
	\caption{Distinguisher used in proof of \thref{thm:pz-new}}
	\label{fig:pz-proof}
	\hrulefill
	\end{figure}
		
	By \thref{o2h} we get $\abs{\Adv^{\dist}_{P,P',e}(\advD)}\leq 2q\sqrt{\Adv^{\guess}_{P,P'}(\advD)}$.
	Note that $\advD$ perfectly simulated the view of $\advA$.
	This holds because $\advD$'s oracle $O_{\Leftrightarrow}$ on input $(X,Y : H_0, H_1)$ returns $(X, Y \xor H_b(X) : H_0, H_1)$ and hence $\advD$ returns $(X, Y \xor H_b(X))$ to $\advA$.
	Note that $\advA$ expects access to $H_b[\xor]$ which provides the identical behavior of mapping $(X,Y)$ to $(X, Y \xor H_b(X))$.
	Hence $\Adv^{\dist}_{P,P',e}(\advD) = \Adv^{\pzdist}_{\env,e}(\advA)$.
	
	It remains to compute $\Adv^{\guess}_{P,P'}(\advD)$.
	Let $I$ denote a random variable taking the value of $i$ sampled inside of $\sG^{\guess}_{P,P'}(\advD)$.
	Define $Q_0=1$, recursively define $Q_{j+1} = Q_j + q_j + 1$, and define the intervals $R_j=\sett{Q_j,\dots,Q_{j+1}-1}$.
	Note that if $I\in R_j$, then $\advD$ was halted when $\advA_j$ made its $I-Q_j+1$-th query.
	Then,
		\begin{align*}
			\Adv^{\guess}_{P,P'}(\advD)
			=\Pr[\sG^{\guess}_{P,P'}(\advD)]
			=\sum_{j=0}^n \Pr[I \in R_j]\Pr[\sG^{\guess}_{P,P'}(\advD) | I \in R_j]
			=\sum_{j=0}^n (q_j/q)\cdot \Adv^{\pzguess}_{\env,j}(\advA).
		\end{align*}
	This completes the proof. \qed 
\end{proof}

\subsection{Comparing the Pan-Zeng and O2H-based Theorems}\label{sec:pz-comparisons}
By depending on $q_j$ and $\Adv^{\pzguess}_{\env,i}(\advA)$, the bounds in \thref{thm:pz} and \thref{thm:pz-new} can be hard to parse. 
When applying \thref{thm:pz}, Pan and Zeng simplified as follows.
Let $\epsilon = \max_i \Adv^{\pzguess}_{\env,i}(\advA)$ and note $q_i\leq q$.
Then,
\begin{align*}
	\sum_{k=0}^n \sum_{i=0}^k 2 q_i \sqrt{\Adv^{\pzguess}_{\env,i}(\advA)}
	\leq \sum_{k=0}^n \sum_{i=0}^k 2 q \sqrt{\epsilon}
	\leq 2(n+1)^2q\sqrt{\epsilon}.
\end{align*}
We can do a little better by noting $q=\sum_{i=0}^n q_i$ and calculating,
\begin{align*}
	\sum_{k=0}^n \sum_{i=0}^k 2 q_i \sqrt{\Adv^{\pzguess}_{\env,i}(\advA)}
	\leq \sum_{i=0}^n \sum_{k=i}^n 2 q_i \sqrt{\epsilon}
	\leq 2(n+1)q\sqrt{\epsilon}.
\end{align*}
Performing similar simplifications to the bound from \thref{thm:pz-new} we get the improved result that 
\begin{align*}
	2 q \sqrt{\sum_{i=0}^n \frac{q_i}{q} \Adv^{\pzguess}_{\env,i}(\advA)}
	\leq 2 q \sqrt{\sum_{i=0}^n \frac{q_i}{q} \epsilon}
	\leq 2 q \sqrt{\epsilon}.
\end{align*}

This analysis implicitly shows that $\Adv^{\guess}_{P,P'}(\advD) \leq \epsilon$ for the $\advD$ defined in our proof.
We might as well then have stuck with the bound $\abs{\Adv^{\pzdist}_{\env,e}(\advA)} \leq 2q\sqrt{\Adv^{\guess}_{P,P'}(\advD)}$ in the theorem where the latter term could have been expressed in terms of measuring one of $\advA$'s $q$ queries chosen uniformly at random.
Our result is for $e\in\sett{1,1/2}$ where that of Pan and Zeng only applied for $e = 1$.

%% file: ghhm.tex

\section{GHHM Adaptive Reprogramming Framework}\label{sec:ghhm}
In this section we prove that (a version of) the Grilo, H{\"o}velmanns, H{\"u}lsing, Majenz (GHHM) framework for tight adaptive reprogramming~\cite{AC:GHHM21} is implied by the Fixed-Permutation O2H (\thref{o2h}).
We recall their setting in \secref{sec:ghhm-frame} and discuss why their result seems not to imply or be implied by that of Pan and Zeng~\cite{PKC:PanZen24}.
In \secref{sec:ghhm-implied}, we state and prove our variant of the framework.
Our bound is looser than that of GHHM in general, but in \secref{sec:ghhm-comparisons} we apply the same simplifications that GHHM use when applying their result in theorems and show that our theorem provides essentially the same concrete bounds in this case.

\subsection{GHHM Framework and Security Theorem}\label{sec:ghhm-frame}
The GHHM framework for tight adaptive program can syntactically be captured by the $\sG^{\pzdist}$ game from \figref{fig:pz-games}.
Let $(\mathsf{Init},\mathsf{Orac},\mathsf{Repro})=\env$ be an environment as follows.
\begin{itemize}
	\item $\mathsf{Init}$ outputs $(s,x,H_0,H_1)$ where $H_0$ and $H_1$ are two copies of the same truly random functions from $\mathcal{X}$ to $\mathcal{Y}$. We will assume without loss of generality that $\mathcal{X}=\bits^l$ and $\mathcal{Y}=\bits^{m}$. Outputs $s$ and $x$ are empty strings.
	\item $\mathsf{Orac}$ interprets its input $y$ as specifying a probability distribution $p$ over $\mathcal{X}$ and samples $x\getsr p$ and $y\getsr \mathcal{Y}$, then outputs $(x,a)=(x,(x,y))$.\footnote{The most general version of GHHM's framework has $p$ output a second ``side-information'' string $x'$. They already showed (Appendix A in the current ePrint version) that the version we use implies this stronger version with side information.}
	\item $\mathsf{Repro}$ given $a=(x,y)$ and $H_1$ outputs a function defined identically to $H_1$ except it maps $x$ to $y$. 
\end{itemize}
Call such an environment a ``GHHM environment''.
Then GHHM proved the following information theoretic result (their Theorem 1).

\begin{theorem}[GHHM Adaptive Reprogramming, \cite{AC:GHHM21}]\label{thm:ghhm}
	Let $\env$ be a GHHM environment and $\advA$ be an adversary for which $\advA_i$ makes at most $q_i$ oracle queries.
	Then
	\[
		\abs{\Adv^{\pzdist}_{\env,1}(\advA)}
		\leq
		\sum_{i=1}^n \left(\sqrt{\hat{q}_i p_{i,\mathrm{max}}} + 0.5\hat{q}_i p_{i,\mathrm{max}}\right).
	\]
	Here $\hat{q}_i=q_0+\dots+q_{i-1}$ and $p_{i,\mathrm{max}} = \mathop{\mathbb{E}}[\max_{x\in\mathcal{X}} p_i(x)]$. The expectation is over the behavior of the game up until the $i$-th oracle query and $p_i$ is a random variable denoting the probability distribution chosen by $\advA_i$.
\end{theorem}

Even though we expressed this result using the same game from the Pan and Zeng framework, it's not clear that \thref{thm:pz} or \thref{thm:ghhm} implies the other. 
The GHHM result seems weaker because it requires reprogrammed points to be information-theoretically (rather than computationally) hidden and requires the initial functions and reprogrammed output to be uniform.
On the other hand, the Pan and Zeng result gives a bound in terms of the probability that $\advA$ finds a point where $H_1$ differs from $H_0$.
To get a meaningful bound we need this probability to be small, so the reprogrammed point must be hard to predict even after the reprogramming occurred.
GHHM (by $\mathsf{Orac}$ giving $x$ to $\advA$) is explicitly not such a setting.
It instead works for cases where the reprogrammed point is hard to predict ahead of time.
Pan and Zeng stated that their result ``generalizes and improves'' the GHHM result.
We notified them that this does not seem to be the case, and they acknowledged that their current result does not capture that of GHHM.

Notably Pan and Zeng are able to apply their result to analyze the selective-opening security of an encryption scheme, a result that a priori would seem like guessing the reprogrammed point should be easy afterwards.
They achieve this by choosing an order of programming that differs from what one's first instinct would use.
In essence, the ``trick'' they use is that in the settings they consider, they can predict ahead of time a small set $S$ such that only points in $S$ will ever be programmed. 
Then they arrange ahead of time for the initial functions $H_0$ and $H_1$ to be different on all of these points and for later reprogramming to switch the functions to being consistent on the reprogrammed point.
It does not seem possible in general to capture the GHHM setting in this manner.

\subsection{The GHHM Theorem is Implied by O2H}\label{sec:ghhm-implied}
Now we state and prove a variant of \thref{thm:pz} which follows from the O2H result \thref{o2h}.
\begin{theorem}\label{thm:ghhm-new}
	Let $\env$ be a GHHM environment and $\advA$ be an adversary for which $\advA_i$ makes at most $q_i$ oracle queries.
	Let $q=q_0+\dots+q_n$.
	Define $\hat{q}_i$ and $p_{i,\mathrm{max}}$ as in \thref{thm:ghhm}.
	Let $e\in\sett{1,1/2}$.
	Then
	\[
		\abs{\Adv^{\pzdist}_{\env,e}(\advA)}
		\leq
		2n\sqrt{\sum_{i=1}^n\hat{q_i} p_{i,\mathrm{max}} / n}.
	\]
\end{theorem}

There are two challenges that make it surprising this theorem can be proven from the Fixed-Permutation O2H theorem. 
The first is that all of the dependence on $q$ is inside the square root, whereas O2H gives us a bound with $q$ outside of the square root. 
Secondly, (and similarly to our comparison with Pan and Zeng's theorem) after a reprogramming the adversary is told the point $x$ at which the oracle was redefined.
Consequently, we can only hope to rely on the hardness of querying the oracle on input $x$ before it was reprogrammed. 
But because the distribution for choosing $x$ is not fixed ahead of time, if we tried to naively apply O2H it's not clear how to make the permutations differ on this unknown $x$ ahead of time. 

The same insights tackle both of these challenges. 
Rather than thinking of the ``bad event'' happening during queries to $H$, we are going to think of the reprogramming process as being performed inside an oracle and the ``bad event'' is querying that oracle on inputs that make it reprogram $H$ on places where it has ``already been queried''. 
Note that ``already been queried'' is not a well defined notion because the adversary could have queried all of $H$ in superposition.
To formalize this idea, we use the techniques of Zhandry~\cite{C:Zhandry19} to represent the quantum accessible random oracle $H$ as a ``sparse'' table which is only non-zero on a few entries that ``have been queried'' by the attacker.
(More precisely, it will be a superposition over such tables.)
Then reprogramming at point $x$ can be viewed as a fixed permutation which swaps $H(x)$ with an auxiliary all-zero register.
We can apply the Fixed-Permutation O2H to compare that reprogramming permutation and a permutation which leaves $H$ untouched. 
Because the other register is all-zero, the permutations will only differ on inputs where the random $x$ happens to equal one of the few non-zero entries of $H$.

In \secref{sec:abkh}, we recall more recent variants of this result for reprogramming of random permutations~\cite{EC:ABKM22,EC:ABKMS24}.
Because no permutation analog for Zhandry's result is known, we are unable to show they follow from Fixed-Permutation O2H.

\begin{proof}
Our first step will be moving to a setting where queries to the hash and requests for reprogramming are both quantum queries to oracles.
Consider the game $\sG^b_0$ for $b\in\bits$ defined in \figref{fig:ghhm-proof-games} which we will use to emulate the execution of $\sG^{\pzdist}_{\env,b}(\advA)$ for the given $\advA$ and $\env$.
The behavior of $\env$ is embedded into the oracles.
The game stores ahead of time the list of random outputs $Z$ that will be programmed into new locations of the hash function and random strings $R$ that will seed sampling the programmed point from the distributions chosen by the attacker.
It uses the following quantum registers.
\begin{itemize}
	\item $W$: The local work space of $\advA$. Its final output guess is obtained by measuring $W[1]$.
	\item $X,Y$: The registers intended for $\advA$ to provide input and receive output, respectively, from its oracles.
	\item $H$: The register storing the table specifying the random oracle.
	\item $R$: The registers storing randomness used to seed the choice of $x$'s according to distributions $p$.
	\item $Z$: The registers storing random strings to be programmed into $H$.
	\item $I$: The register storing a counter tracking how many reprogramming queries have occurred. It determines which entries of $R$ and $Z$ are used. 
\end{itemize}
We think of $\advB$ having direct access to the registers $W,X,Y$.
All other registers are controlled exclusively by the game and may not directly be modified by $\advB$.
Note that both oracles of the game are defined by classical permutations.
It is not strictly necessary to purify all of the variables in this way. 
We do so out of an aesthetic preference to have all variables be quantum, rather than having a mix of quantum and classical values, so that we can represent each oracle as single permutation queried in superposition. 
We refer the interested reader to Liao, Ge, and Xue~\cite[Thm.~7]{EPRINT:LiaJiaXue25} for a version of the proof that leaves more of the variables classical.

Given $\advA=(\advA_0,\dots,\advA_n)$ we can map it to an algorithm $\advB$ playing $\sG^b_0$ as follows.
It initially runs on $\advA_0$ on input $\varepsilon$.
Whenever it queries $X,Y$ to its hash oracle, $\advB$ forwards this to its $\procRo$ oracle.
When $\advA_i$ would end a stage, outputting probability distributions $p$, $\advB$ puts the representation of this into $X$ and prepares $Y$ on in the all-zero state.
It makes a $\procReprogram_b$ query and returns the result to the next phase of $\advA$.
When $\advA_n$ halts with output $y$, $\advB$ stores that in $W[1]$ and halts.

\ifnum\eprintversion=0
        \renewcommand{\lwid}{0.4}
        \renewcommand{\mwid}{}
		\renewcommand{\rwid}{0.45}
\else
        \renewcommand{\lwid}{0.21}
        \renewcommand{\mwid}{}
        \renewcommand{\rwid}{0.27}
\fi
\begin{figure}[t]
\twoColsNoDivide{\lwid}{\rwid}
{
\underline{Games $\sG^b_0(\advB)$}\\
$\ro\getsr\Func(l,m)$\\
$\textbf{Z}\getsr\Func(\lceil\lg q\rceil, m)$\\
$\textbf{R}\getsr\Func(\lceil\lg q\rceil, \infty)$\\
$\ket{H,Z,R}\gets\ket{\ro,\mathbf{Z},\mathbf{R}}$\\
Run $\advB[{\procRo,\procReprogram_b}]$\\
Measure $W[1]$\\
Return $W[1]=1$\\[4pt]
\underline{$\procRo(X,Y : H)$}\\
$Y \gets Y \xor H[X]$\\
Return $(X,Y : H)$\\[4pt]
\underline{$\procFRo(X,Y : H)$}\\
$H[X] \gets H[X] \xor Y$\\
Return $(X,Y : H)$
\smallskip
}
{
\underline{Games $\sG^b_1(\advB)$}\\
$\textbf{R}\getsr\Func(\lceil\lg q\rceil, \infty)$\\
$\ket{R}\gets\ket{\mathbf{R}}$\\
Run $\advB[\had^Y\circ\procFRo\circ\had^Y,\procReprogram_b]$\\
Measure $W[1]$\\
Return $W[1]=1$\\[8pt]
\underline{$\procReprogram_b(X,Y:H,I,Z,R)$ }\\
Interpret $X$ as prob.~dist.~$p$\\
$x\gets p(R[I])$\\
If $b=1$ then\\
\ind $(H[x],Z[I])\gets(Z[I],H[x])$\\
$Y \gets Y \xor x$\\
$I\gets I+1\bmod 2^{\lceil\lg q\rceil}$\\ 
Return $(X,Y:H,I,Z,R)$\smallskip
}
\caption{Games used in the analysis of \thref{thm:ghhm-new}. Registers not explicitly initialized are initialized to all $0$.}
\label{fig:ghhm-proof-games}
\hrulefill
\end{figure}
	
When $b=1$, each reprogram oracle query programs a fresh random output (from $Z$) into the location of $H$ chosen according to $p$.
When $b=0$, nothing is ever programmed into $H$.
These match the behaviors of $H_1$ and $H_0$ in $\sG^{\pzdist}$ so $\advB$ perfectly simulates the view of $\advA$, giving $\Adv^{\pzdist}_{\env,e}(\advA) = (\Pr[\sG^1_0(\advB)])^e - (\Pr[\sG^0_0(\advB)])^e$.

Next we will transition to the games $\sG^b_1$.
In these games, rather than being sampled at random, the registers $H$ and $Z$ are initialized to zeros.
The oracle $\procRo$ which writes $H[X]$ into $Y$ has instead been replaced with $\had^Y\circ\procFRo\circ\had^Y$ here $\had^Y$ is the Hadamard transform applied to register $Y$ and $\procFRo$ is the Fourier version of $\procRo$ where instead $Y$ is written into $H[X]$.

By the sparse QROM representation technique of Zhandry, these games are completely identical so $\Adv^{\pzdist}_{\env,e}(\advA) = (\Pr[\sG^1_1(\advB)])^e - (\Pr[\sG^0_1(\advB)])^e$ and the view of $\advA$ inside of $\advB$ is unchanged.
In \apref{app:ghhm-details} we break this claim into smaller atomic steps for those unfamiliar with the technique.

Now we can compare the behavior of the permutations $\procReprogram_1$ and $\procReprogram_0$.
They only differ in whether $H[x]$ and $Z[I]$ are swapped ($b=1$) or not ($b=0$).
Thus $\procReprogram_1$ and $\procReprogram_0$ will only differ on inputs for which $H[x]\neq Z[I]$.

We bound the difference between the $b=1$ and $b=0$ worlds of the game using the Fixed-Permutation O2H (\thref{o2h}). 
Let $P=\procReprogram_0$, $P'=\procReprogram_1$, and $\advD$ be the distinguisher that runs all of $\sG^b_1$ internally, except it calls its oracle on $(X,Y:H,I,Z,R)$ whenever $\advB$ calls its reprogramming oracle on $X,Y$.
Note that it makes $n$ oracle queries.
We have that $\Adv^{\pzdist}_{\env,e}(\advA) = \Adv^{\dist}_{P,P',e}(\advD)$. 

We complete the proof by analyzing $\Adv^{\guess}_{P,P'}(\advD)$, though leaving the bound in terms of $\Adv^{\guess}_{P,P'}(\advD)$ would make it applicable to settings where the reprogrammed points are only computationally hard to predict.
Consider a fixed choice of $i\in\sett{1,\dots,n}$ in $\sG^{\guess}_{P,P'}(\advD)$.
The permutations differ on the measured input iff $H[x]\neq Z[i-1]$ (because $I$ will have value $i$).
By construction of the game $Z[i-1]$ is necessarily all zeros so this is equivalent to asking whether $H[x]$ is zero.
Note that $H$ is independent of $R[i-1]$ so we can think of $H$ being measured first, then $R[i-1]$ being measured.
So $x=p(R[i-1])$ looks like a fresh sample from $p$.
At this point, $\advB$ inside of $\advD$ has made at most $\hat{q_i}$ queries to $\procFRo$ so $H$ is a superposition of tables each of which has at most $\hat{q_i}$ non-zero entries.
The result of measuring $H$ will consequently be such a table.
The probability of a non-zero entry being hit by a fresh sample from $p$ is at most $\hat{q_i}\cdot p_{i,\mathrm{max}}$ from a union bound over the non-zero entries of the measured value of $H$.
Averaging over the possible choice of $i$ we have
\[
	\Adv^{\guess}_{P,P'}(\advD) = \sum_{i=1}^n \Pr[i]\Pr[\sG^{\guess}_{P,P'}(\advD) | i] \leq \sum_{i=1}^n (1/n)\hat{q_i}\cdot p_{i,\mathrm{max}}.
\]
Applying \thref{o2h} gives the claimed bound\qed
\end{proof}

\subsection{Comparing the GHHM and O2H-based Theorems}\label{sec:ghhm-comparisons}

\heading{Simplified Comparison}
By depending on $\hat{q}_j$ and $p_{i,\mathrm{max}}$, the bounds in \thref{thm:ghhm} and \thref{thm:ghhm-new} can be hard to parse. 
Consequently, when applying \thref{thm:ghhm} in proofs, GHMM simplified as follows (e.g., from their proposition 2).
Let $p_{\mathrm{max}} = \max_i p_{i,\mathrm{max}}$. 
Note that $\hat{q}_i\leq q$, and $qp_{\mathrm{max}} \leq \sqrt{qp_{\mathrm{max}}}$.
Then,
\[
	\sum_{i=1}^n \left(\sqrt{\hat{q}_i p_{i,\mathrm{max}}} + 0.5\hat{q}_i p_{i,\mathrm{max}}\right)
	\leq
	\sum_{i=1}^n 1.5 \sqrt{q p_{\mathrm{max}}}
	=
	1.5n\sqrt{q p_{\mathrm{max}}}.
\]
Performing analogous simplifications from the bound in \thref{thm:ghhm-new} gives
\[
	2n\sqrt{\sum_{i=1}^n\hat{q_i} p_{i,\mathrm{max}} / n}
	\leq
	2n\sqrt{\sum_{i=1}^n qp_{\mathrm{max}} / n}
	=
	2n\sqrt{qp_{\mathrm{max}}}.
\]
Our result is for $e\in\sett{1,1/2}$ where that of GHHM only applied for $e = 1$.

\heading{General Comparison}
Liao, Ge, and Xue~\cite{EPRINT:LiaJiaXue25} point out that the Cauchy-Schwartz inequality shows the bound in \thref{thm:ghhm-new} is in general worse than that of \thref{thm:ghhm}.
Cauchy-Schwartz says that $(\sum_{i=1}^n x_iy_i)^2 \leq (\sum_{i=1}^n x_i^2)(\sum_{i=1}^n y_i^2)$ for real $x_i,y_i$, which simplifies to $(\sum_{i=1}^n x_i)^2 \leq n(\sum_{i=1}^n x_i^2)$ when $y_i=1$.
Now define $x_i=\sqrt{\hat{q}_i p_{i,\mathrm{max}}}$ and assume $\hat{q}_i p_{i,\mathrm{max}}\leq 1$.
Then,
\[
	\sum_{i=1}^n \left(\sqrt{\hat{q}_i p_{i,\mathrm{max}}} + 0.5\hat{q}_i p_{i,\mathrm{max}}\right)
	\leq 1.5 \sum_{i=1}^n x_i
	\leq 1.5\sqrt{n\left(\sum_{i=1}^n x_i^2\right)}
	= 1.5 n\sqrt{\sum_{i=1}^n \hat{q}_i p_{i,\mathrm{max}} / n}.
\]
By replacing our use of the Fixed-Permutation O2H with their new Double-Oracle Fixed-Permutation O2H, Liao, Ge, and Xue~\cite[Thm.~7]{EPRINT:LiaJiaXue25} prove an improved bound for $e\in\sett{1,1/2}$ of 
\[
	\abs{\Adv^{\pzdist}_{\env,e}(\advA)}
	\leq
	\sqrt{8 \sum_{i=1}^n \hat{q}_i p_{i,\mathrm{max}}}.
\]

%% file: unruh.tex

\section{Unruh's Adaptive O2H}\label{sec:unruh}
In this section, we show that Unruh's Adaptive O2H's results~\cite{C:Unruh14,EC:Unruh15} (which generalize an earlier result Unruh~\cite{unruh}) are implied by the Fixed-Permutation O2H (\thref{o2h}). 
To proving the stronger version~\cite{EC:Unruh15} we emulate the ideas from our proof in \secref{sec:ghhm}.

\subsection{First Adaptive O2H}\label{sec:unruh-one}
The first adaptive O2H result we analyze bounds how well an adversary can distinguish between $H(x,m)$ and a random string where the attacker chooses $m$ and $x$ is uniformly random. 
This is defined by the games shown in \figref{fig:unruh} for which we define 
\[
	\Adv^{\undist}_{e}(\advA)=\left(\Pr[\sG^{\undist}_1(\advA)]\right)^e-\left(\Pr[\sG^{\undist}_0(\advA)]\right)^e.
\]

The bound will consist of two terms.
Intuitively, the first information theoretically bounds the how much $\advA_0$ can contribute the advantage because $x$ is independent of its view.
The second terms bounds in terms of how likely $\advA_1$ is to query $(x,m)$ to its oracle, captured formally by $\Adv^{\unguess}_{}(\advA)=\Pr[\sG^{\unguess}_{}(\advA)]$.

\ifnum\eprintversion=0
        \renewcommand{\lwid}{0.36}
        \renewcommand{\mwid}{}
        \renewcommand{\rwid}{0.49}
\else
        \renewcommand{\lwid}{0.23}
        \renewcommand{\mwid}{}
        \renewcommand{\rwid}{0.36}
\fi
\begin{figure}[t]
\twoColsNoDivide{\lwid}{\rwid}
{
\underline{Game $\sG^{\undist}_{b}(\advA)$}\\
$H\getsr\Func(l+k,n)$\\
$m \getsr \advA_0[H[\oplus]]$\\
$x \getsr \bits^{l}$\\ 
$B_0 \gets H(x,m)$\\
$B_1 \getsr \bits^n$\\
$b'\getsr\advA_1[H[\oplus]](x,B_b)$\\
Return $b'=1$\smallskip
}
{
\underline{Game $\sG^{\unguess}(\advA)$}\\
$i\getsr\sett{1,\dots,q_1}$\\
$H\getsr\Func(l+k,n)$\\
$m \getsr \advA_0[H[\oplus]]$\\
$x \getsr \bits^{l}$\\ 
$B_1 \getsr \bits^n$\\
Run $\advA_1[H[\oplus]](x,B_1)$ until its $i$-th query\\
Measure the input $(x',m')$ to this query\\
Return $(x,m)=(x',m')$\smallskip
}
\caption{Games used for Unruh's adaptive O2H~\cite{C:Unruh14}. Different stages of $\advA$ implicitly share state.}
\label{fig:unruh}
\hrulefill
\end{figure}

Of these advantages, Unruh proves the following relationship (Lemma 14 in the current ePrint version). 
\begin{theorem}[Unruh Adaptive O2H, \cite{C:Unruh14}]\label{thm:unruk}\label{thm:unruh}
	Let $\advA=(\advA_0,\advA_1)$ be an adversary for which $\advA_i$ makes at most $q_i$ oracle queries.
	Then
	\[
		\abs{\Adv^{\undist}_{1}(\advA)} 
		\leq
		q_0 2^{-l/2+2}
		+
		2q_1\sqrt{\Adv^{\unguess}_{}(\advA)}.	
	\]
\end{theorem}

We prove this result for $\Adv^{\undist}_{e}(\advA)$ with $e\in\sett{1,1/2}$ from the Fixed-Permutation O2H (\thref{o2h}).
Starting from the real world, the proof will first use O2H to program $\advA_0$'s oracle to be different on \emph{all} inputs starting with $x$.
Then $H(x,m)$ will only ever be used to as input to $\advA_1$ and in response to its oracle queries.
We switch both uses to use $B_1$ instead, which is equivalent. 
Then we apply O2H again to switch $\advA_1$'s oracle back to using $H(x,m)$.

In \secref{sec:unruh2} our proof of Unruh's second result will actually imply this theorem with a better concrete bound, replacing $q_0 2^{-l/2+2}$ with $\sqrt{q_0}2^{-l/2+1}$.
We find it pedagogically useful to start with this proof first.
Our proof of Unruh's second result will apply the ideas we used for proving GHHM's result in \secref{sec:ghhm} --- in particular, using a sparse representation of the random oracle so that we can check for whether $\advA_0$ ``queried'' $x$ only after it has stopped executing.

AHU~\cite{C:AmbHamUnr19} used their O2H theorem to prove post-quantum security of a Fujisaki-Okamoto variant --- a result previously proven by using Unruh's adaptive O2H result~\cite{TCC:TarUnr16}.\footnote{In fact, both the proofs had flaws (see the ePrint version of~\cite{C:AmbHamUnr19}). This is orthogonal to our discussion here.}
While discussing the differences AHU say,
\begin{quote}
While our
O2H Theorem is not adaptive (in the sense that the input where the oracle is
reprogrammed has to be fixed at the beginning of the game), it turns out that
in the present case our new O2H Theorem can replace the adaptive one. This is
because our new O2H Theorem allows us to reprogram the oracle at a large number
of inputs (not just a single one). It turns out we do not need to adaptively choose
the one input to reprogram, we just reprogram all potential inputs. At least in the
proof from~\cite{TCC:TarUnr16}, this works without problems. 
\end{quote}
Our proof uses this idea of reprogramming the oracle at a large number of points, showing that the O2H theorem can replace \thref{thm:unruk} in \emph{any} proof, not just the one from~\cite{C:AmbHamUnr19}.
Our stronger proof in \secref{sec:unruh2} will only require reprogramming at a single point. 

	\begin{figure}[t]
	\threeColsNoDivide{0.21}{0.27}{0.27}
	{
	\underline{Hybrids $\sH_{(a,b,c)}$}\\
	$H\getsr\Func(l+k,n)$\\
	$h\getsr\Func(k,n)$\\
	$x \getsr \bits^{l}$\\ 
	$m \getsr \advA_0[\procRo^0_{a}]$\\
	$B \gets H(x,m)$\comment{$b=0$}\\
	$B_1 \getsr \bits^n$\\
	$B\gets B_1$\comment{$b=1$}\\
	$b'\getsr\advA_1[\procRo^1_{c}](x,B)$\\
	Return $b'=1$\smallskip
	}
	{
	\underline{$\procRo^{0}_{a}(X,M,Y:H,h,x)$}\\
	If $X=x$ then\\
	\ind $Y \gets Y \xor H[X,M]$\comment{$a=0$}\\
	\ind $Y \gets Y \xor h[M]$\comment{$a=1$}\\
	Else $Y \gets Y \xor H[X,M]$\\
	Return $(X,M,Y:H,h,x)$\smallskip
	}
	{
	\underline{$\procRo^1_{c}(X,M,Y:H,B_1,x,m)$}\\
	If $(X,M)=(x,m)$ then\\
	\ind $Y \gets Y \xor H[X,M]$\comment{$c=0$}\\
	\ind $Y \gets Y \xor B_1$\comment{$c=1$}\\
	Else $Y \gets Y \xor H[X,M]$\\
	Return $(X,M,Y:H,B_1,x,m)$
	\smallskip
	}
	\caption{Hybrid games used for proof of \thref{thm:unruk}}
	\label{fig:unruh-pf}
	\hrulefill
	\end{figure}

\begin{proof}
	Our proof will use the hybrid games $\sH_{(a,b,c)}$ shown in \figref{fig:unruh-pf} which are parameterized by $(a,b,c)\in\bits^{3}$.
	Informally, we will show the following.
	\[
		\sG^{\dist}_{0} 
		\equiv
		\sH_{(0,0,0)}
		\underset{{2q_0\sqrt{2^{-l}}}}{\approx}
		\sH_{(1,0,0)}
		\equiv
		\sH_{(1,1,1)}
		\underset{{2q_0\sqrt{2^{-l}}}}{\approx}
		\sH_{(0,1,1)}
		\underset{{2q_1\sqrt{\varepsilon^{\mathsf{ow}}}}}{\approx}
		\sH_{(0,1,0)}
		\equiv
		\sG^{\dist}_{1}. 
	\]  
	We proceed from left to right.
	Compared to $\sG^{\undist}_0$, hybrid $\sH_{(0,0,0)}$ samples an additional (unused) random function $h$ and samples $x$ at the beginning of the game.
	All three parameters being zero means both random oracles respond correctly using $H$ and that $\advA_1$ is given $H(x,m)$ as input. 
	So game $\sG^{\undist}_0$ is equivalent to game $\sH_{(0,0,0)}$.
	
	Now compare $\sH_{(0,0,0)}$ to $\sH_{(1,0,0)}$.
	The change to $a$ means that $\advA_0$'s oracle will return $h(M)$ on any input of the form $(x,M)$.
	Consider a distinguisher $\advD$ for the Fixed-Permutation O2H (\thref{o2h}) which runs all of $\sH_{(?,0,0)}$ except it forwards $\advA_0$'s oracle queries to its own oracle which is either $P=\procRo^0_0$ or $P'=\procRo^0_1$.
	The adversary provides $X,M,Y$ while $\advD$ provides the rest of the inputs. 
	It outputs the same bit $\advA_1$ does.
	Clearly $\advD$ correctly simulates the view of $\advA$.
		
	Note that the permutations only differ on inputs for which $X=x$ and that $\advA_0$'s view is independent of $x$ when interacting with $\procRo^0_0$.
	Consequently, $\Adv^{\guess}_{P,P'}(\advD)\leq 1/2^l$ and so 
	\(\abs{\left(\Pr[\sH_{(0,0,0)}]\right)^e 
		-
		\left(\Pr[\sH_{(1,0,0)}]\right)^e}
		\leq 
		2q_0\sqrt{2^{-l}}.\)

	In $\sH_{(1,0,0)}$, $\advA_1$ is given the random string $H(x,m)$ as input and then its oracle returns $H(x,m)$ on input $(x,m)$.
	In $\sH_{(1,1,1)}$, $\advA_1$ is given the random string $B_1$ as input and then its oracle returns $B_1$ on input $(x,m)$.
	Note that these variables are otherwise unused because $\advA_0$ cannot access $H(x,m)$.
	Consequently the games are perfectly equivalent.
	Then we can move back to $\sH_{(0,1,1)}$ with analysis analogous to the transition from $\sH_{(0,0,0)}$ to $\sH_{(1,0,0)}$. 
	Putting these steps together gives \(\abs{\left(\Pr[\sH_{(1,0,0)}]\right)^e 
		-
		\left(\Pr[\sH_{(0,1,1)}]\right)^e}
		\leq 
		2q_0\sqrt{2^{-l}}.\)

	Now comparing $\sH_{(0,1,1)}$ and $\sH_{(0,1,0)}$ we see that they differ only in the behavior of $\advA_1$'s oracle.
	The former will return $B_1$ on input $(x,m)$ while the latter will return $H(x,m)$.
	Consider a distinguisher $\advD'$ for the Fixed-Permutation O2H (\thref{o2h}) which runs all of $\sH_{(0,1,?)}$ except it forwards $\advA_1$'s oracle queries to its own oracle which is either $P=\procRo^1_1$ or $P'=\procRo^1_0$.
	The adversary provides $X,M,Y$ while $\advD$ provides the rest of the inputs. 
	It outputs the same bit $\advA_1$ does.
	Clearly $\advD$ correctly simulates the view of $\advA$.
	
	We have \(\abs{\left(\Pr[\sH_{(0,1,1)}]\right)^e 
		-
		\left(\Pr[\sH_{(0,1,0)}]\right)^e} \leq 2q_1\sqrt{\Adv^{\guess}_{P,P'}(\advD')}\).
	Note that the permutations only differ on inputs for which $(X,M)=(x,m)$ and that the view of $\advA$ run by $\advD'$ in $\sG^{\guess}_{P,P'}$ matches the view it would get in $\sG^{\unguess}$.
	Consequently, $\Adv^{\guess}_{P,P'}(\advD') \leq \Adv^{\unguess}_{}(\advA)$.
	
	Finally, we can compare $\sH_{(0,1,0)}$ and $\sG^{\undist}_0$ to see that they are equivalent. 
	Putting together our claims and using the triangle inequality gives the bound
	\[
		\abs{\Adv^{\undist}_{e}(\advA)}
		\leq
		4q_0\sqrt{2^{-l}}
		+
		2q_1\sqrt{\Adv^{\guess}_{P,P'}(\advD')}
		.\]
	This completes the proof. \qed
\end{proof}

\subsection{Second Adaptive O2H}\label{sec:unruh2}
In~\cite{EC:Unruh15}, Unruh improved on his adaptive O2H result with a version that allowed the hidden point to be chosen according to a arbitrary adaptively chosen distribution, as long as this distribution has sufficient entropy.
We can capture the result using games $\sG^{\undist2}$ and $\sG^{\unguess2}$ defined in \figref{fig:unruh2}.
In both, sampling $x$ is now done with the classical algorithm $\advA_C$. 
It cannot access $\advA_0$'s state beyond the input $m$ passed to it. 
Algorithm $\advA_1$ is allowed allow access the state of both $\advA_0$ and $\advA_1$. 
Defining
$\Adv^{\undist2}_{e}(\advA)=\left(\Pr[\sG^{\undist2}_1(\advA)]\right)^e-\left(\Pr[\sG^{\undist2}_0(\advA)]\right)^e$ and $\Adv^{\unguess2}_{}(\advA)=\Pr[\sG^{\unguess2}_{}(\advA)]$.

\begin{figure}[t]
\twoColsNoDivide{0.22}{0.35}
{
\underline{Game $\sG^{\undist2}_{b}(\advA)$}\\
$H\getsr\Func(l,n)$\\
$m \getsr \advA_0[H[\oplus]]$\\
$x \getsr \advA_C(m)$\\ 
$B_0 \gets H(x)$\\
$B_1 \getsr \bits^n$\\
$b'\getsr\advA_1[H[\oplus]](x,B_b)$\\
Return $b'=1$\smallskip
}
{
\underline{Game $\sG^{\unguess2}(\advA)$}\\
$i\getsr\sett{1,\dots,q_1}$\\
$H\getsr\Func(l,n)$\\
$m \getsr \advA_0[H[\oplus]]$\\
$x \getsr \advA_C(m)$\\ 
$B_1 \getsr \bits^n$\\
Run $\advA_1[H[\oplus]](x,B_1)$ until its $i$-th query\\
Measure the input $x'$ to this query\\
Return $x=x'$\smallskip
}
\caption{Games used for Unruh's second adaptive O2H~\cite{EC:Unruh15}. Algorithm $\advA_C$ is classical. Algorithm $\advA_1$ can access the final state of $\advA_0$ and $\advA_C$.}
\label{fig:unruh2}
\hrulefill
\end{figure}

We define the collision entropy $k$ and min-entropy $\mu$ of $\advA_C$ by
\begin{align*}
	k &= \min_m -\log_2\Pr[x=y~:~x\getsr\advA_C(m), y\getsr\advA_C(m)]\\
	\mu	&= \min_{m,x} -\log_2\Pr[x=y~:~y\getsr\advA_C(m)].
\end{align*}
Note that $2\mu\geq k \geq \mu$.

Then we can state Unruh's result (Lemma 9 in the current ePrint version).
\begin{theorem}[Unruh Adaptive O2H, \cite{EC:Unruh15}]\label{thm:unruh2}
	Let $\advA=(\advA_0,\advA_1)$ be an adversary for which $\advA_i$ makes at most $q_i$ oracle queries and $\advA_C$ be a classical algorithm with collision entropy $k$.
	Then
	\[
		\abs{\Adv^{\undist2}_{1}(\advA)} 
		\leq
		2q_1\sqrt{\Adv^{\unguess2}_{}(\advA)}
		+
		(4+\sqrt{2})
		\sqrt{q_0} 2^{-k/4}.	
	\]
\end{theorem}

We prove the following slightly generalized result.
\begin{theorem}\label{thm:unruh2-ours}
	Let $\advA=(\advA_0,\advA_1)$ be an adversary for which $\advA_i$ makes at most $q_i$ oracle queries and $\advA_C$ be a classical algorithm with min-entropy $\mu$ and collision entropy $k$.
	Then for $e\in\sett{1,1/2}$.
	\begin{align*}
		\abs{\Adv^{\undist2}_{e}(\advA)} 
		&\leq
		2q_1\sqrt{\Adv^{\unguess2}_{}(\advA)}
		+
		2
		\sqrt{q_0} 2^{-\mu/2}\\
		&\leq
		2q_1\sqrt{\Adv^{\unguess2}_{}(\advA)}
		+
		2
		\sqrt{q_0} 2^{-k/4}.	
	\end{align*}
\end{theorem}
Unruh conjectured that the $2^{-k/4}$ factor is an artifact of their proof technique.
If so, our proof technique has the same artifact but removes it when min-entropy is an acceptable replacement. 
This, for example, allows us to directly imply a version of \thref{thm:unruh} with the bound improved to replace $q_0$ with $\sqrt{q_0}$.

Our proof combines the ideas from our proof of Unruh's first adaptive O2H (\thref{thm:unruk}) and our proof of the GHHM adaptive reprogramming result (\thref{thm:ghhm}).
We first move to a hybrid where we reprogram $H(x)$ to equal $B_1$ before $\advA_1$ is executed (more precisely, we swap the values of $H(x)$ and $B_1$).
To obtain a tighter bound for this step than in \thref{thm:unruk}, we switch to a sparse representation of $H$ and bound the ``bad event'' at the time of the swap, as in our \thref{thm:ghhm} proof. 
Then we show the difference between this hybrid and the final game by reprogramming $\advA_1$'s access to $H$ on input $x$.

\begin{figure}[t]
\twoColsNoDivide{0.26}{0.33}
{
\underline{Games $\sG_b(\advA)$}\\
$\textbf{H}\getsr\Func(l,n)$\\
$\textbf{B}\getsr\bits^n$\\
$\textbf{R}\getsr\bits^\infty$\\
$\ket{H,B,R}\gets\ket{\mathbf{H},\mathbf{B},\mathbf{R}}$\\
Run $\advA_0[\procRo]$\\
Run $\procSamp_b$ \\
Run $\advA_1[\procRo]$\\
Measure $W[1]$\\
Return $W[1]=1$\smallskip
}
{
\underline{$\procSamp_b(X,Y:H,B,V,R,X^{\ast})$ }\\
$V \gets V \xor X$\\
$X^{\ast}\gets X^{\ast} \xor \advA_C(X;R[I])$\\
$Y \gets Y \xor H[X^{\ast}]$\comment{If $b=0$}\\
$Y \gets Y \xor B$\comment{If $b=1$}\\
Return $(X,Y:H,B,V,R,X^{\ast})$\\[4pt]
\underline{$\procRo(X,Y : H)$}\\
$Y \gets Y \xor H[X]$\\
Return $(X,Y : H)$
}
\twoColsNoDivide{0.26}{0.33}
{
\underline{Hybrids $\sH_{(a,b,c)}$}\\
$\textbf{R}\getsr\Func(\lceil\lg q\rceil, m)$\\
$\ket{R}\gets\ket{\mathbf{R}}$\\
Run $\advA_0[\had^Y\circ\procFRo^0\circ\had^Y]$\\
Run $\had^Y\circ\procFSamp_{a,b}\circ\had^Y$ \\
Run $\advA_1[\had^Y\circ\procFRo^1_c\circ\had^Y]$\\
Measure $W[1]$\\
Return $W[1]=1$\\[4pt]
\underline{$\procFRo^0(X,Y : H)$}\\
$H[X] \gets Y \xor H[X]$\\
Return $(X,Y : H)$\smallskip
}
{
\underline{$\procFSamp_{a,b}(X,Y:H,B,V,R,X^{\ast})$ }\\
$V \gets V \xor X$\\
$X^{\ast}\gets X^{\ast} \xor \advA_C(X;R[I])$\\
$(B,H[X^{\ast}])\gets (H[X^{\ast}],B)$\comment{If $a=1$}\\
$H[X^{\ast}] \gets Y \xor H[X^{\ast}]$\comment{If $b=0$}\\
$B \gets Y \xor B$\comment{If $b=1$}\\
Return $(X,Y:H,B,V,R,X^{\ast})$\\[4pt]
\underline{$\procFRo^1_c(X,Y : H, B)$}\\
If $X=X^{\ast}$ then\\
\ind $H[X^{\ast}] \gets Y \xor H[X^{\ast}]$\comment{If $c=0$}\\
\ind $B \gets Y \xor B$\comment{If $c=1$}\\
Else $H[X] \gets Y \xor H[X]$\\
Return $(X,Y : H, B)$\smallskip
}
\caption{Hybrid games for proof of \thref{thm:unruh2-ours}, implying the adaptive O2H result of  Unruh~\cite{EC:Unruh15}. Algorithm $\advA_1$, but not $\advA_0$, may access register $R$.}
\label{fig:unruh-pf2}
\hrulefill
\end{figure}

\begin{proof}
	For this proof we use the games shown in \figref{fig:unruh-pf2}.
	We start with $\sG_b$ which is simply a rewritten version of $\sG^{\undist2}_{b}$.
	It writes the use of $\advA_C$ as a single query to an oracle $\procSamp$ and make everything quantum.
	Algorithm $\advA_0$ acts on registers $W,X,Y$, while $\advA_1$ may additionally act on $R$.
	(Here $R$ stores the randomness that will be used by $\advA_C$. Giving it to $\advA_1$ is equivalent for our purposes to letting $\advA_1$ access the final state of $\advA_C$ and additionally allows $\advA_1$ to recompute the $x$ which is stored in $X^{\ast}$.)
	Registers $H$ and $B$ are controlled by the game. 
	To enforce that the $\procSamp$ query is classical, it writes the query $X$ (which stores $m$ at this time) into the otherwise unused register $V$.
	So we have $\Pr[\sG^{\undist2}_b(\advA)]=\Pr[\sG_b(\advA)]$.
	Formally, the syntax of $\advA$ changed, so on the right it should have been replaced with an appropriately defined $\advA'$.
	
	Our proof will use the hybrid games $\sH_{(a,b,c)}$ shown in \figref{fig:unruh-pf2} which are parameterized by $(a,b,c)\in\bits^{3}$.
	Informally, we will show the following.
	\[
		\sG_{0} 
		\equiv
		\sH_{(0,0,0)}
		\underset{{2\sqrt{q_02^{-\mu}}}}{\approx}
		\sH_{(1,0,0)}
		\equiv
		\sH_{(0,1,1)}
		\underset{{2q_1\sqrt{\varepsilon^{\mathsf{ow}}}}}{\approx}
		\sH_{(0,1,0)}
		\equiv
		\sG_{1}. 
	\]  
	We proceed from left to right.
	Consider $\sH_{(0,0,0)}$, 
	In this game, rather than being sampled at random, the registers $H$ and $B$ are initialized to zeros.
	The oracle $\procRo$ which writes $H[X]$ into $Y$ has instead been replaced with two (for now equivalent) oracles $\had^Y\circ\procFRo^i\circ\had^Y$ where $\had^Y$ is the Hadamard transform applied to register $Y$ and $\procFRo$ is the Fourier version of $\procRo$ where instead $Y$ is written into $H[X]$.
	Similarly, $\procFSamp$ which writes $H[X^\ast]$ or $B$ into $Y$ has been replaced with $\had^Y\circ\procSamp_{a,b}\circ\had^Y$ which writes $Y$ into $H[X^\ast]$ or $B$.
	By the sparse QROM representation technique of Zhandry, these games are completely equivalent, giving $\Pr[\sG_b(\advA)]=\Pr[\sH_{(0,b,0)}(\advA)]$.
	See our proof of \thref{thm:ghhm-new} and in particular \apref{app:ghhm-details} for an example of how to break this equivalence claims into smaller atomic steps.
	
	Now in $\sH_{(0,1,0)}$ we perform an additional swap of $B$ and $H[X^*]$.
	Note that $\procFSamp_{a,0}$ and $\procFSamp_{a,1}$ differ as permutations only if $B\neq H[X^*]$.
	
	We apply the Fixed-Permutation O2H (\thref{o2h}) with $\advD$ that runs all of $\sH_{(0,?,0)}$ internally, except it calls its oracle to emulate $\procFSamp$. 
	Note that $\advD$ makes one oracle query and at that time $H$ has at most $q_0$ non-zero entries (which thus differ from $B$ which is zero).
	Thus $\Adv^{\guess}_{\procFSamp_{0,0},\procFSamp_{1,0}}(\advD)\leq q_0/2^{-\mu}$ from the min-entropy of $\advA_C$ and $|(\Pr[\sH_{(0,0,0)}(\advA)])^e - (\Pr[\sH_{(0,1,0)}(\advA)])^e|\leq 2\sqrt{q_0/2^{\mu}}$.

	Hybrid $\sH_{(0,1,0)}$ swaps the registers $B$ and $H[X^*]$ before $\advA_1$ is run.
	Note $\advA_1$ does not have direct access to these registers, so it would be equivalent to leave $B$ and $H[X^*]$ unswapped, but instead switch which of the two is used in all future accesses to the registers.
	The is what's done in the equivalent game $\sH_{0,1,1}$.
	
	Now $\sH_{(0,1,1)}$ differs from $\sH_{(0,1,0)}$ only when $X^*$ is queried to $\procFRo^1$.
	Apply the Fixed-Permutation O2H (\thref{o2h}) with $\advD'$ that runs all of $\sH_{(0,1,?)}$ internally except it uses its own oracle to respond to $\procFRo^1$.
	We get $|(\Pr[\sH_{(0,1,0)}(\advA)])^e - (\Pr[\sH_{(0,1,1)}(\advA)])^e|\leq 2q_1\sqrt{\Adv^{\guess}_{\procFRo^1_{1},\procFRo^1_{0}}(\advD')}$.
	
	The previously discussed equivalence between $\sH_{(0,1,0)}$, $\sG_1$, and $\sG^{\undist2}_1(\advA)$ allows us to conclude that $\Adv^{\guess}_{\procFRo^1_{1},\procFRo^1_{0}}(\advD')=\Adv^{\unguess2}_{}(\advA)$ and complete the proof using the triangle inequality. \qed
\end{proof}

%% file: other-frameworks.tex

\section{Limitations of Fixed-Permutation O2H}\label{sec:abkh}
\begin{wrapfigure}{r}{0.33\textwidth}\small
\vspace{-4em}
\oneCol{0.26}
{
\underline{Game $\sG^{\abkhdist}_b(\advA)$}\\
$\Pi_0\getsr\Perm(n)$\\
$\advA_0[\Pi_0[\xor],\Pi^{-1}_0[\xor]]$\\
$s,s'\getsr\bits^n$\\
$\Pi_1 \gets \Pi_0 \circ S_{s,s'}$\\
$b'\getsr\advA_1[\Pi_b[\xor],\Pi^{-1}_b[\xor]](s,s')$\\
Return $b'=1$\smallskip
}
\caption{Game for ABKH resampling for permutations result}
\label{fig:abkh}
\end{wrapfigure}
Alagic, Bai, Katz, and Majenz~\cite{EC:ABKM22} gave a result for the resampling of random permutations which they call an extension of the GHHM adaptive reprogramming lemma (\thref{thm:ghhm}) to the case of two-way accessible random permutations which we cannot reproduce from Fixed-Permutation O2H because it is not known if quantumly accessible random permutations can be sparsely represented.

Define the swap permutation $S_{s,s'}$ by $S(s)=s'$, $S(s')=s$, and $S(x)=x$ otherwise.
Let $\Perm(n)$ denote the set of all permutations on $\bits^n$.
Then the relevant security game is shown in \figref{fig:abkh} for which we define $\Adv^{\abkhdist}_e(\advA)=\left(\Pr[\sG^{\abkhdist}_1(\advA)]\right)^e-\left(\Pr[\sG^{\abkhdist}_0(\advA)]\right)^e$.

Of this game, they prove the following bound (their Lemma 5).
\begin{theorem}[ABKH Permutation Resampling, \cite{EC:ABKM22}]\label{thm:abkh}
	Let $\advA=(\advA_0,\advA_1)$ be an adversary for which $\advA_0$ makes at most $q_0$ oracle queries.
	Then
	\[
		\abs{\Adv^{\abkhdist}_{1}(\advA)}
		\leq
		4\sqrt{q_0/2^n}.
	\]
\end{theorem}

While this result is intuitively related to the GHHM adaptive reprogramming result, we cannot port over our proof from \secref{sec:ghhm}, because it relied on Zhandry's technique for sparse representation of a random oracle.
No sufficiently analogous technique is currently known for random permutations.
(Indeed it is a notoriously difficult problem~\cite{AC:Unruh23}.)

\iftrue
\else
Recall that there were two reasons we relied on Zhandry's technique in that section.
First because the reprogramming points were sampled from an adaptively chosen distribution and then immediately given to the attacker we could only rely on the ``bad event'' of querying these points before the reprogramming occurred.
But at that time we don't know what points are bad!
That is not an issue for this theorem, as $s$ and $s'$ could have been sampled at the beginning of the game.

The second reason is an issue. 
We wanted bounds with the number of oracle queries inside the square root (note \thref{thm:abkh} has this form).
This was obtained by thinking of reprogramming being its own oracle and the ``bad event'' being that the chosen reprogramming point happened to coincide with a previous oracle query, formally a non-zero entry in the sparse representation of the random oracle.
Without a sparse representation for random permutations, it is unclear how to emulate this.
\fi 

Below for comparison, we prove a bound of the form $2q_0\sqrt{2/2^n}$ for \thref{thm:abkh} using Fixed-Permutation O2H by exploiting the fact that $s,s'$ are nonadaptively chosen.

In essence, the issues here correspond to the differences between the two proofs in \secref{sec:unruh}.
The proof we provide here is analogous to our proof of \thref{thm:unruk}, and we are unable to prove an analog of \thref{thm:unruh2} because we lack techniques for sparsely representing permutations.

By replacing our use of the Fixed-Permutation O2H with their new Double-Oracle Fixed-Permutation O2H, Liao, Ge, and Xue~\cite[Thm.~9]{EPRINT:LiaJiaXue25} prove that $\abs{\Adv^{\abkhdist}_{e}(\advA)}\leq 4\sqrt{q_0/2^n}$ for $e\in\sett{1,1/2}$. 

\begin{proof}[Weakened \thref{thm:abkh}]
	First note that $s,s'\getsr\bits^n$ could have been sampled at the beginning of the game.
	Then it would have been completely equivalent to give $\advA_0$ access to $\Pi_b$ and $\advA_1$ access to $\Pi_0$. 
	We can define the fixed permutations $\procPerm_b(X,Y, d: \Pi, s, s')$ to implement $\Pi_b$ if $d=1$ and $\Pi^{-1}_b$ if $d=0$. 
	They differ only on inputs for which $(X,d) \in \sett{(s,1),(s',1),(\Pi_0(s),0),(\Pi(s'),0)}$.
	So applying \thref{o2h} with these two permutations and $\advD$ that picks $s,s'$, runs $\advA_0$ with access to the permutation, then internally simulates the rest of the game and outputs whatever $\advA$ does we get 
	\[
		\Adv^{\abkhdist}_{e}(\advA)
		=
		\Adv^{\dist}_{\procPerm_1,\procPerm_0,e}(\advD)
		\leq
		2q_0\sqrt{\Adv^{\guess}_{\procPerm_1,\procPerm_0}(\advD)}
		=
		2q_0\sqrt{2/2^n}.
	\] 
	The last equality comes from noting the view of $\advA_0$ is independent of $s,s'$ and that for a given choice of $d$ there are at most two uniformly random $X$'s on which the permutations differ.\qed
\end{proof}

Alagic, Bai, Katz, Majenz, and Struck~\cite{EC:ABKMS24} introduced a stronger version of this theorem which allows $s$ to be sampled according to an adaptively chosen distribution with high min-entropy ($s'$ is still uniform).
For this version, the proof technique above would not work because we cannot sample $s$ at the beginning of the game.
At best, we could prove a variant in which $\advA_1$ is never told $s,s'$ and so we can bound its success based its ability to query the permutations on the bad points after they are chosen.

%% file: ghhm-details.tex

\ifnum\eprintversion=0
        \renewcommand{\lwid}{0.4}
        \renewcommand{\mwid}{}
		\renewcommand{\rwid}{0.45}
\else
        \renewcommand{\lwid}{0.39}
        \renewcommand{\mwid}{}
        \renewcommand{\rwid}{0.525}
\fi
\begin{figure}[h]
\begin{minipage}{0.5\textwidth}
\begin{flushleft}\setlength{\fboxsep}{2pt}
        \framebox{
        \begin{tabular}{l@{\hspace*{.2em}}r}
        \begin{minipage}[t]{\lwid\textwidth}
        	\gamesfontsize
			\underline{Games $\sG^b_0(\advB)$}\\
			$\ro\getsr\Func(l,m)$\\
			$\textbf{Z}\getsr\Func(\lceil\lg q\rceil, m)$\\
			$\textbf{R}\getsr\Func(\lceil\lg q\rceil, \infty)$\\
			$\ket{H,Z,R}\gets\ket{\ro,\mathbf{Z},\mathbf{R}}$\\
			Run $\advB[{\procRo,\procReprogram_b}]$\\
			Measure $W[1]$\\
			Return $W[1]=1$\\[4pt]
			\underline{$\procRo(X,Y : H)$}\\
			$Y \gets Y \xor H[X]$\\
			Return $(X,Y : H)$\\[4pt]
			\underline{$\procFRo(X,Y : H)$}\\
			$H[X] \gets H[X] \xor Y$\\
			Return $(X,Y : H)$\smallskip
        \end{minipage}
        &
        \begin{minipage}[t]{\rwid\textwidth}
        	\gamesfontsize
			\underline{Games $\sG^b_1(\advB)$}\\
			$\textbf{R}\getsr\Func(\lceil\lg q\rceil, m)$\\
			$\ket{R}\gets\ket{\mathbf{R}}$\\
			Run $\advB[\had^Y\circ\procFRo\circ\had^Y,\procReprogram_b]$\\
			Measure $W[1]$\\
			Return $W[1]=1$\\[8pt]
			\underline{$\procReprogram_b(X,Y:H,I,Z,R)$ }\\
			Interpret $X$ as prob.~dist.~$p$\\
			$x\gets p(R[I])$\\
			If $b=1$ then\\
			\ind $(H[x],Z[I])\gets(Z[I],H[x])$\\
			$Y \gets Y \xor x$\\
			$I\gets I+1\bmod 2^{\lceil\lg q\rceil}$\\ 
			Return $(X,Y:H,I,Z,R)$\smallskip
		\end{minipage}
        \end{tabular}
        }
\end{flushleft}

\end{minipage}
\begin{minipage}{0.5\textwidth}
	\begin{flushleft}
	        \framebox{\setlength{\fboxsep}{1pt}
	        \begin{tabular}{c}
	        \begin{minipage}[t]{0.935\textwidth}
	        	\gamesfontsize
	        	\underline{Hybrid $\sH^b_\kappa$}\\
				$\textbf{R}\getsr\Func(\lceil\lg q\rceil, \infty)$\\
				$\ket{R}\gets\ket{\mathbf{R}}$\\
				$\ket{H,Z}\gets\had\ket{H,Z}$ \comment{$\sH_1$,$\sH_2$}\\
				Run $\advB[\procRo,\procReprogram_b]$ \comment{$\sH_1$}\\
				Run $\advB[\had^{Y,H}\circ\procFRo\circ\had^{Y,H},\had^{Z,H}\circ\procReprogram_b\circ\had^{Z,H}]$ \comment{$\sH_2$}\\
				Run $\advB[\had^{Y}\circ\procFRo\circ\had^{Y},\procReprogram_b]$ \comment{$\sH_3$}\\
				$\ket{H,Z}\gets\had\ket{H,Z}$ \comment{$\sH_3$}\\
				Measure $W[1]$\\
				Return $W[1]=1$\small
	        \end{minipage}
	        \end{tabular}
	        }
	\end{flushleft}
\end{minipage}
\caption{\textbf{Left:} Reproduction of games from the proof of \thref{thm:ghhm-new}. \textbf{Right:} Hybrid games for using Zhandry's sparse representation technique. Registers not explicitly initialized are initialized to all $0$.}
\label{fig:ghhm-proof-games-2}
\label{fig:zha-use}
\hrulefill
\end{figure}

\section{Details of Sparse QROM Representation in \thref{thm:ghhm-new}}\label{app:ghhm-details}
We analyze why the games $\sG^b_0$ and $\sG^b_1$ as defined in \figref{fig:ghhm-proof-games-2} (reproduced from \figref{fig:ghhm-proof-games} in the proof of \thref{thm:ghhm-new}) are equivalent.
We will consider a sequence of hybrid games $\sH^b_1$ through $\sH^b_3$ as defined in \figref{fig:zha-use} for which $\sG^b_0$ is equivalent to $\sH_1$, $\sG^b_1$ is equivalent to $\sH^b_3$, and $\sH_\kappa$ is equivalent to $\sH_{\kappa+1}$ for each $\kappa$.

\heading{Hybrid 1}
First we define $\sH^b_1$ identically to $\sG^b_0$ except that $H$ and $Z$ are initialized by applying the Hadamard transform to the all zeros strings.
Recall that $\had\ket{x} = 1/\sqrt{2^n} \cdot \sum_{x'} (-1)^{x\cdot x'} \ket{x'}$ if $x$ is a bitstring $x\in\bits^n$.
When $x$ is the all zeros string, this gives the uniform superposition $1/\sqrt{2^n} \cdot \sum_{x'} \ket{x'}$.
Thus, if we measured $H$ and $Z$ immediately, we would be assigning them a uniformly random function as in $\sG^b_0$.
Values in $H$ and $Z$ are only ever swapped with each other or used to control xor's into other registers.
So by the principle of joint deferred measurement, leaving them unmeasured is equivalent.

\heading{Hybrid 2}
Next we define $\sH^b_2$ identically to $\sH^b_1$ except the oracles $\procRo$ and $\procReprogram_b$ have been replaced with $\had^{Y,H}\circ\procFRo\circ\had^{Y,H}$ and $\had^{Z,H}\circ\procReprogram_b\circ\had^{Z,H}$.
These do not change the behavior of the game because $\procRo=\had^{Y,H}\circ\procFRo\circ\had^{Y,H}$ and $\procReprogram_b = \had^{Z,H}\circ\procReprogram_b\circ\had^{Z,H}$.
These equalities follow from the following lemma (and the fact that $\had$ is its own inverse).
\begin{lemma}
	Define permutations $P(x,y) = (x\xor y, y)$, $P'(x,y) = (x, x\xor y)$, $S(x,y)=(y,x)$, and $I(x,y)=(x,y)$. Let $U$ be an arbitrary unitary with inverse $U^{-1}$. Then
	\begin{align*}
		P = \had \circ P' \circ \had,
		&&
		S = (U \otimes U) \circ S \circ (U^{-1} \otimes U^{-1}),
		&&
		\text{ and}
		&&
		I = (U \otimes U) \circ I \circ (U^{-1} \otimes U^{-1}).
	\end{align*}
\end{lemma}

\heading{Hybrid 3}
We define $\sH^b_3$ by cancelling out Hadamard transforms in $\sH_2$, using that $\had$ is its own inverse.
Note that $\had^H$ commutes with $\advB$ because $\advB$ does not act on register $H$.
Thereby, the initial $\had^{H}$ used to setup $H$ cancels with the $\had^H$ before the first oracle call.
The $\had^H$'s between any two oracle queries cancel with each other.
This leaves the $\had^H$ after the last oracle query, which $\sH^b_3$ defers to the end of the game.

Similarly, neither $\advB$ nor $\procFRo$ act on register $Z$ so transform $\had^Z$ will commute with them.
This similarly allows us to cancel out all $\had^Z$ except the one after the last query which is deferred until the end of the game.

We cannot cancel $\had^Y$ similarly, because $\advB$ acts on register $Y$.

Finally, comparing $\sH^b_3$ with $\sG^b_1$ we see they are identical except that $\sH^b_3$ has a $\ket{H,Z}\gets\had\ket{H,Z}$ at the end.
At that point of execution all that matters is the measurement of $W[1]$ which is unaffected by this operation, so it can be removed.

%% file: main.bbl
\begin{thebibliography}{10}
\providecommand{\url}[1]{\texttt{#1}}
\providecommand{\urlprefix}{URL }
\providecommand{\doi}[1]{https://doi.org/#1}

\bibitem{EC:ABKM22}
Alagic, G., Bai, C., Katz, J., Majenz, C.: Post-quantum security of the
  {Even}-{Mansour} cipher. In: Dunkelman, O., Dziembowski, S. (eds.)
  EUROCRYPT~2022, Part~III. {LNCS}, vol. 13277, pp. 458--487. Springer, Cham
  (May~/~Jun 2022). \doi{10.1007/978-3-031-07082-2_17}

\bibitem{EC:ABKMS24}
Alagic, G., Bai, C., Katz, J., Majenz, C., Struck, P.: Post-quantum security of
  tweakable {Even}-{Mansour}, and applications. In: Joye, M., Leander, G.
  (eds.) EUROCRYPT~2024, Part~I. {LNCS}, vol. 14651, pp. 310--338. Springer,
  Cham (May 2024). \doi{10.1007/978-3-031-58716-0_11}

\bibitem{C:AmbHamUnr19}
Ambainis, A., Hamburg, M., Unruh, D.: Quantum security proofs using
  semi-classical oracles. In: Boldyreva, A., Micciancio, D. (eds.) CRYPTO~2019,
  Part~II. {LNCS}, vol. 11693, pp. 269--295. Springer, Cham (Aug 2019).
  \doi{10.1007/978-3-030-26951-7_10}

\bibitem{CCS:BelRog93}
Bellare, M., Rogaway, P.: Random oracles are practical: {A} paradigm for
  designing efficient protocols. In: Denning, D.E., Pyle, R., Ganesan, R.,
  Sandhu, R.S., Ashby, V. (eds.) ACM CCS 93. pp. 62--73. {ACM} Press (Nov
  1993). \doi{10.1145/168588.168596}

\bibitem{EC:BelRog06}
Bellare, M., Rogaway, P.: The security of triple encryption and a framework for
  code-based game-playing proofs. In: Vaudenay, S. (ed.) EUROCRYPT~2006.
  {LNCS}, vol.~4004, pp. 409--426. Springer, Berlin, Heidelberg (May~/~Jun
  2006). \doi{10.1007/11761679_25}

\bibitem{TCC:BHHHP19}
Bindel, N., Hamburg, M., H{\"o}velmanns, K., H{\"u}lsing, A., Persichetti, E.:
  Tighter proofs of {CCA} security in the quantum random oracle model. In:
  Hofheinz, D., Rosen, A. (eds.) TCC~2019, Part~II. {LNCS}, vol. 11892, pp.
  61--90. Springer, Cham (Dec 2019). \doi{10.1007/978-3-030-36033-7_3}

\bibitem{AC:BDFLSZ11}
Boneh, D., Dagdelen, {\"O}., Fischlin, M., Lehmann, A., Schaffner, C., Zhandry,
  M.: Random oracles in a quantum world. In: Lee, D.H., Wang, X. (eds.)
  ASIACRYPT~2011. {LNCS}, vol.~7073, pp. 41--69. Springer, Berlin, Heidelberg
  (Dec 2011). \doi{10.1007/978-3-642-25385-0_3}

\bibitem{EC:CheSte14}
Chen, S., Steinberger, J.P.: Tight security bounds for key-alternating ciphers.
  In: Nguyen, P.Q., Oswald, E. (eds.) EUROCRYPT~2014. {LNCS}, vol.~8441, pp.
  327--350. Springer, Berlin, Heidelberg (May 2014).
  \doi{10.1007/978-3-642-55220-5_19}

\bibitem{C:DFPS23}
Devevey, J., Fallahpour, P., Passel{\`e}gue, A., Stehl{\'e}, D.: A detailed
  analysis of {Fiat}-{Shamir} with aborts. In: Handschuh, H., Lysyanskaya, A.
  (eds.) CRYPTO~2023, Part~V. {LNCS}, vol. 14085, pp. 327--357. Springer, Cham
  (Aug 2023). \doi{10.1007/978-3-031-38554-4_11}

\bibitem{C:DFHS24}
Don, J., Fehr, S., Huang, Y.H., Struck, P.: On the (in)security of the {BUFF}
  transform. In: Reyzin, L., Stebila, D. (eds.) CRYPTO~2024, Part~I. {LNCS},
  vol. 14920, pp. 246--275. Springer, Cham (Aug 2024).
  \doi{10.1007/978-3-031-68376-3_8}

\bibitem{SAC:Eaton17}
Eaton, E.: {Leighton}-{Micali} hash-based signatures in the quantum
  random-oracle model. In: Adams, C., Camenisch, J. (eds.) SAC 2017. {LNCS},
  vol. 10719, pp. 263--280. Springer, Cham (Aug 2017).
  \doi{10.1007/978-3-319-72565-9_13}

\bibitem{AC:GHHM21}
Grilo, A.B., H{\"o}velmanns, K., H{\"u}lsing, A., Majenz, C.: Tight adaptive
  reprogramming in the {QROM}. In: Tibouchi, M., Wang, H. (eds.)
  ASIACRYPT~2021, Part~I. {LNCS}, vol. 13090, pp. 637--667. Springer, Cham (Dec
  2021). \doi{10.1007/978-3-030-92062-3_22}

\bibitem{PKC:HulRijSon16}
H{\"u}lsing, A., Rijneveld, J., Song, F.: Mitigating multi-target attacks in
  hash-based signatures. In: Cheng, C.M., Chung, K.M., Persiano, G., Yang, B.Y.
  (eds.) PKC~2016, Part~I. {LNCS}, vol.~9614, pp. 387--416. Springer, Berlin,
  Heidelberg (Mar 2016). \doi{10.1007/978-3-662-49384-7_15}

\bibitem{TCC:JaeSonTes21}
Jaeger, J., Song, F., Tessaro, S.: Quantum key-length extension. In: Nissim,
  K., Waters, B. (eds.) TCC~2021, Part~I. {LNCS}, vol. 13042, pp. 209--239.
  Springer, Cham (Nov 2021). \doi{10.1007/978-3-030-90459-3_8}

\bibitem{C:JZCWM18}
Jiang, H., Zhang, Z., Chen, L., Wang, H., Ma, Z.: {IND}-{CCA}-secure key
  encapsulation mechanism in the quantum random oracle model, revisited. In:
  Shacham, H., Boldyreva, A. (eds.) CRYPTO~2018, Part~III. {LNCS}, vol. 10993,
  pp. 96--125. Springer, Cham (Aug 2018). \doi{10.1007/978-3-319-96878-0_4}

\bibitem{PKC:KosXag24}
Kosuge, H., Xagawa, K.: Probabilistic hash-and-sign with retry in the quantum
  random oracle model. In: Tang, Q., Teague, V. (eds.) PKC~2024, Part~I.
  {LNCS}, vol. 14601, pp. 259--288. Springer, Cham (Apr 2024).
  \doi{10.1007/978-3-031-57718-5_9}

\bibitem{EC:KSSSS20}
Kuchta, V., Sakzad, A., Stehl{\'e}, D., Steinfeld, R., Sun, S.:
  Measure-rewind-measure: Tighter quantum random oracle model proofs for
  one-way to hiding and {CCA} security. In: Canteaut, A., Ishai, Y. (eds.)
  EUROCRYPT~2020, Part~III. {LNCS}, vol. 12107, pp. 703--728. Springer, Cham
  (May 2020). \doi{10.1007/978-3-030-45727-3_24}

\bibitem{EPRINT:LiaJiaXue25}
Liao, H., Ge, J., Xue, R.: Non-adaptive one-way to hiding not only implies
  adaptive quantum reprogramming, but also does better. Cryptology ePrint
  Archive (2025), \url{https://eprint.iacr.org/2025/1998}

\bibitem{EC:Maurer02}
Maurer, U.M.: Indistinguishability of random systems. In: Knudsen, L.R. (ed.)
  EUROCRYPT~2002. {LNCS}, vol.~2332, pp. 110--132. Springer, Berlin, Heidelberg
  (Apr~/~May 2002). \doi{10.1007/3-540-46035-7_8}

\bibitem{AC:MorYam22}
Morimae, T., Yamakawa, T.: Classically verifiable {NIZK} for {QMA} with
  preprocessing. In: Agrawal, S., Lin, D. (eds.) ASIACRYPT~2022, Part~IV.
  {LNCS}, vol. 13794, pp. 599--627. Springer, Cham (Dec 2022).
  \doi{10.1007/978-3-031-22972-5_21}

\bibitem{PKC:PanZen24}
Pan, J., Zeng, R.: Selective opening security in the quantum random oracle
  model, revisited. In: Tang, Q., Teague, V. (eds.) PKC~2024, Part~II. {LNCS},
  vol. 14603, pp. 92--122. Springer, Cham (Apr 2024).
  \doi{10.1007/978-3-031-57725-3_4}

\bibitem{SAC:Patarin08}
Patarin, J.: The ``coefficients {H}'' technique (invited talk). In: Avanzi,
  R.M., Keliher, L., Sica, F. (eds.) SAC 2008. {LNCS}, vol.~5381, pp. 328--345.
  Springer, Berlin, Heidelberg (Aug 2009). \doi{10.1007/978-3-642-04159-4_21}

\bibitem{EPRINT:Shoup04}
Shoup, V.: Sequences of games: a tool for taming complexity in security proofs.
  Cryptology ePrint Archive, Report 2004/332 (2004),
  \url{https://eprint.iacr.org/2004/332}

\bibitem{TCC:TarUnr16}
Targhi, E.E., Unruh, D.: Post-quantum security of the {Fujisaki}-{Okamoto} and
  {OAEP} transforms. In: Hirt, M., Smith, A.D. (eds.) TCC~2016-B, Part~II.
  {LNCS}, vol.~9986, pp. 192--216. Springer, Berlin, Heidelberg (Oct~/~Nov
  2016). \doi{10.1007/978-3-662-53644-5_8}

\bibitem{C:Unruh14}
Unruh, D.: Quantum position verification in the random oracle model. In: Garay,
  J.A., Gennaro, R. (eds.) CRYPTO~2014, Part~II. {LNCS}, vol.~8617, pp. 1--18.
  Springer, Berlin, Heidelberg (Aug 2014). \doi{10.1007/978-3-662-44381-1_1}

\bibitem{EC:Unruh15}
Unruh, D.: Non-interactive zero-knowledge proofs in the quantum random oracle
  model. In: Oswald, E., Fischlin, M. (eds.) EUROCRYPT~2015, Part~II. {LNCS},
  vol.~9057, pp. 755--784. Springer, Berlin, Heidelberg (Apr 2015).
  \doi{10.1007/978-3-662-46803-6_25}

\bibitem{unruh}
Unruh, D.: Revocable quantum timed-release encryption. J. ACM  \textbf{62}(6)
  (Dec 2015). \doi{10.1145/2817206}

\bibitem{AC:Unruh23}
Unruh, D.: Towards compressed permutation oracles. In: Guo, J., Steinfeld, R.
  (eds.) ASIACRYPT~2023, Part~IV. {LNCS}, vol. 14441, pp. 369--400. Springer,
  Singapore (Dec 2023). \doi{10.1007/978-981-99-8730-6_12}

\bibitem{ACISP:YuaSunTak24}
Yuan, Q., Sun, C., Takagi, T.: Revisiting the security of fiat-shamir signature
  schemes under superposition attacks. In: Zhu, T., Li, Y. (eds.) ACISP 24,
  Part~II. {LNCS}, vol. 14896, pp. 164--184. Springer, Singapore (Jul 2024).
  \doi{10.1007/978-981-97-5028-3_9}

\bibitem{ACISP:YuaTibAbe23}
Yuan, Q., Tibouchi, M., Abe, M.: Quantum-access security of hash-based
  signature schemes. In: Simpson, L., Baee, M.A.R. (eds.) ACISP 23. {LNCS},
  vol. 13915, pp. 343--380. Springer, Cham (Jul 2023).
  \doi{10.1007/978-3-031-35486-1_16}

\bibitem{FOCS:Zhandry12}
Zhandry, M.: How to construct quantum random functions. In: 53rd FOCS. pp.
  679--687. {IEEE} Computer Society Press (Oct 2012).
  \doi{10.1109/FOCS.2012.37}

\bibitem{C:Zhandry12}
Zhandry, M.: Secure identity-based encryption in the quantum random oracle
  model. In: Safavi-Naini, R., Canetti, R. (eds.) CRYPTO~2012. {LNCS},
  vol.~7417, pp. 758--775. Springer, Berlin, Heidelberg (Aug 2012).
  \doi{10.1007/978-3-642-32009-5_44}

\bibitem{C:Zhandry19}
Zhandry, M.: How to record quantum queries, and applications to quantum
  indifferentiability. In: Boldyreva, A., Micciancio, D. (eds.) CRYPTO~2019,
  Part~II. {LNCS}, vol. 11693, pp. 239--268. Springer, Cham (Aug 2019).
  \doi{10.1007/978-3-030-26951-7_9}

\end{thebibliography}
